\documentclass[11pt,letterpaper]{article}

\usepackage{fullpage}
\usepackage{setspace}

\usepackage{amsmath,amssymb,amsfonts}
\usepackage{xspace}

\usepackage{color}

\usepackage{comment}

\usepackage[lined,boxed,linesnumbered]{algorithm2e}

  \usepackage{theorem}

\newtheorem{fact}{Fact}[section]
\newtheorem{lemma}[fact]{Lemma}
\newtheorem{theorem}[fact]{Theorem}
\newtheorem{definition}[fact]{Definition}

\newtheorem{claim}[fact]{Claim}

\newtheorem{remark}[fact]{Remark}

\newenvironment{proof}{{\bf Proof:  }}{\hfill\rule{2mm}{2mm}}

\newcommand{\tO}{\tilde O}
\newcommand{\polylog}{{\rm polylog}}
\newcommand{\eps}{\varepsilon}
\newcommand{\reals}{{\mathbb R}}
\newcommand{\Nat}{{\mathbb N}}
\newcommand{\opt}{{\sf OPT}}

\newcommand{\preserver}{\textsc{Pairwise Distance Preserver}\xspace}
\newcommand{\pairwise}{\textsc{Pairwise Spanner}\xspace}
\newcommand{\DSF}{\textsc{Directed Steiner Forest}\xspace}
\newcommand{\additive}{\textsc{Additive $k$-Spanner}\xspace}
\newcommand{\basic}{\textsc{Basic $k$-Spanner}\xspace}
\newcommand{\MR}{\textsc{Min-Rep}\xspace}

\newcommand{\classNP}{{\mathrm{NP}}}

\newcommand{\classDTIME}{{\mathrm{DTIME}}}

\allowdisplaybreaks

\newcounter{note}[section]

\newcommand{\group}{\Gamma}

\usepackage{url}

\usepackage{ifpdf}
\ifpdf    
\usepackage{hyperref}
\else    
\usepackage[hypertex]{hyperref}
\fi

\makeatletter
\def\url@leostyle{%
  \@ifundefined{selectfont}{\def\UrlFont{\sf}}{\def\UrlFont{\small\ttfamily}}}
\makeatother
\title{Approximating Spanners and Directed Steiner Forest: Upper and Lower Bounds}
\author{Eden Chlamt\'a\v{c}\thanks{Partially supported by ISF grant 1002/14.}\\Ben Gurion University \and Michael Dinitz\thanks{Supported in part by NSF awards CCF-1464239 and CCF-1535887.}\\Johns Hopkins University \and Guy Kortsarz\thanks{Supported in part by NSF grants 1218620 and 1540547}\\Rutgers University-Camden \and Bundit Laekhanukit\\Weizmann Institute of Science}

\begin{document}

\begin{titlepage}
\def\thepage{}
\maketitle
\begin{abstract}
It was recently found that there are very close connections between the existence of \emph{additive spanners} (subgraphs where all distances are preserved up to an additive stretch), \emph{distance preservers} (subgraphs in which demand pairs have their distance preserved exactly), and \emph{pairwise spanners} (subgraphs in which demand pairs have their distance preserved up to a multiplicative or additive stretch) [Abboud-Godwin SODA '16, Godwin-Williams SODA '16].  We study these problems from an optimization point of view, where rather than studying the existence of extremal instances we are given an instance and are asked to find the sparsest possible spanner/preserver.  We give an $O(n^{3/5 + \eps})$-approximation for distance preservers and pairwise spanners (for arbitrary constant $\eps > 0$).  This is the first nontrivial upper bound for either problem, both of which are known to be as hard to approximate as Label Cover.  We also prove Label Cover hardness for approximating additive spanners, even for the cases of additive $1$ stretch (where one might expect a polylogarithmic approximation, since the related multiplicative $2$-spanner problem admits an $O(\log n)$-approximation) and additive polylogarithmic stretch (where the related multiplicative spanner problem has an $O(1)$-approximation)

Interestingly, the techniques we use in our approximation algorithm extend beyond distance-based problem to pure connectivity network design problems.  In particular, our techniques allow us to give an $O(n^{3/5 + \eps})$-approximation for the Directed Steiner Forest problem (for arbitrary constant $\eps > 0$) when all edges have uniform costs, improving the previous best $O(n^{2/3 + \eps})$-approximation due to Berman et al.~[ICALP '11] (which holds for general edge costs).
\end{abstract}

\end{titlepage}\pagenumbering {arabic} 


\section{Introduction}

There has been significant recent progress on problems involving \emph{graph spanners}: subgraphs which approximately preserve distances.  The traditional notion of spanner has involved multiplicative stretch, in which all distances are preserved up to a multiplicative factor known as the \emph{stretch}.  If this stretch factor is $k$ then the subgraph is known as a \emph{$k$-spanner}.  It has been known for over $20$ years that all undirected graphs admit sparse spanners if the stretch factor is at least $3$, where sparser and sparser spanners are possible as the stretch factor increases.  More formally, Alth\"ofer et al.~\cite{ADDJS93} showed that for every graph $G$ and integer $k \geq 1$, there is a $(2k-1)$-spanner of $G$ with at most $n^{1+1/k}$ edges.  Moreover, this is tight assuming the Erd\H{o}s girth conjecture~\cite{erdos-girth}.

A different but synergistic question involves \emph{optimizing} spanners: given an input graph $G$ and a stretch value $k$, can we algorithmically find the \emph{sparsest} $k$-spanner of $G$?  After all, not all graphs are those from the girth conjectures, and if some graph does have a sparse spanner, we would like to find it.  This is known as the \basic problem~\cite{DKR16}, and similar optimization problems can be defined for many other versions of spanners where tradeoffs do not exist, such as for directed graphs~\cite{BGJRW09,DK11-stoc,BBMRY13,DZ16}, when the objective is to minimize the maximum degree rather than the sparsity~\cite{CDK12,CD14}, etc.  Many of these problems can be thought of as standard network design problems (e.g.~Steiner Tree, Steiner Forest, etc.) but where the connectivity constraint is augmented with a distance constraint: not only do certain nodes need to be connected, the distance of the connecting path must be short.  There has been significant recent progress, both positive and negative, on many of these problems: in general they tend to be as hard to approximate as Label Cover~\cite{AL97} (but not always), but it is still often possible to give nontrivial approximation algorithms.  

In parallel with this work on optimizing spanners, there has been rapid progress on understanding what other tradeoffs are possible.  Three directions in particular have been the focus of much of this work, and in fact have been shown to be related~\cite{AB16-soda,BW16}: \emph{additive spanners}, \emph{distance preservers}, and \emph{pairwise spanners}.  In additive spanners we restrict attention to unweighted graphs, but ask for the stretch to be additive rather than multiplicative (giving us more leeway for small distances, but less flexibility for long distances).  In distance preservers and pairwise spanners we make an orthogonal change to the spanner definition: instead of preserving \emph{all} distances, we are given some subset of pairs of nodes $T \subseteq V \times V$ (known as the demands) and are only required to preserve distances between demand pairs.  In a preserver we must preserve these distances \emph{exactly}, while in a pairwise spanner we again allow some stretch (multiplicative or additive) for the demands.  All three of these objects exhibit somewhat surprising behavior (see Section~\ref{sec:related} for a more detailed discussion), but have also shown to be related to each other (for example, preserver lower bounds can imply pairwise spanner lower bounds~\cite{AB16-soda}, and some overlapping techniques have been useful for both preservers and additive spanners~\cite{BW16}).

In this paper we bring together these lines of work, by studying the approximability of distance preservers, pairwise spanners, and additive spanners.  We provide the first nontrivial upper bounds for distance preservers and pairwise spanners, while for additive spanners we provide the first known hardness of approximation.  Our hardness of approximation in some ways mirrors recent results on tradeoffs for additive spanners~\cite{AB16-stoc}, in that we can provide hardness for settings where the associated multiplicative problem is actually easy.  Moreover, the techniques required for our upper bounds also yield new insight into classical network design problems, giving an improved approximation for the Directed Steiner Forest problem when all edges have the same cost.  

\subsection{Our Results and Techniques} \label{sec:results}

We now give our results, and for each result some short intuition about the techniques used.  Given a graph $G$, we will let $d_G(u,v)$ be the shortest-path distance between $u$ and $v$ in $G$ (note that we are assuming unweighted graphs throughout this paper, so the distance is the same as the number of edges).  If $G$ is clear from context we may omit it, and simply refer to $d(u,v)$.  

\begin{definition}
Given a directed graph $G = (V, E)$ and a collection $P \subseteq V \times V$, a subgraph $H = (V, E')$ is a \emph{pairwise distance preserver} if $d_H(u,v) = d_G(u,v)$ for all $(u,v) \in P$.   
\end{definition}

\begin{definition}
In the \preserver problem we are given a graph $G = (V, E)$ (possibly directed and with edge lengths) and a collection $P \subseteq V \times V$, and are asked to return a pairwise distance preserver $H$ which minimizes $|E(H)|$.  
\end{definition}

Our first result is for distance preservers, where we give the first nontrivial upper bound.  

\begin{theorem} \label{thm:preservers-main}
For any constant $\epsilon > 0$, there is a polynomial-time $O(n^{3/5 + \epsilon})$-approximation algorithm for the \preserver problem.
\end{theorem}

The main difficulty in proving this theorem is the fact that $n$ is not a lower bound on the optimal solution.  Most approximation algorithms for spanners (e.g.~\cite{BGJRW09,DK11-stoc,BBMRY13,DZ16}) involve randomly sampling shortest-path trees: since $n-1$ is a lower bound on $\opt$ for spanners, sampling $f$ trees only costs us $f$ in the approximation ratio.  But for preservers, $n$ is no longer a lower bound since the optimal solution need not be connected.  Hence sampling shortest path trees is no longer low cost, and without this step any existing algorithm for spanners we might apply has unbounded cost.  

To overcome this, we replace shortest-path trees with \emph{junction trees}.  Junction trees are trees which cover ``significant" demand at ``little" cost.  Typically they involve a root node $r$, and a collection of shortest paths into $r$ and shortest paths out of $r$ in order to satisfy some of the demand.  Due to their simple structure we can usually find the densest junction tree (the tree with best ratio of demand covered to cost incurred) relatively efficiently, and then due to their high density these junction trees can be combined into a global solution at low cost.  They have been used extensively in network design since their introduction by~\cite{CHKS10}, but this is the first time (as far as we are aware) that they have been used for spanner problems, or for any problem with a hard distance constraint.  We believe that bringing this technique into spanners is a significant contribution of this work.

By using some of the ideas about junction trees that we developed for the \preserver problem, we can give an improved approximation algorithm for a classic network design problem: \DSF with uniform costs.

\begin{definition}
In the \DSF problem we are given a directed graph $G = (V, E)$, nonnegative edge costs $c: E \rightarrow \mathbb{R}^{\geq 0}$, and a collection of node pairs $P \subseteq V \times V$.  If all edges have the same cost, we say that the edge costs are \emph{uniform}.  We are asked to return a subgraph $H$ of $G$ which minimizes $\sum_{e \in E(H)} c(e)$ subject to there being a directed path from $s$ to $t$ for all $(s,t) \in P$.
\end{definition}

\begin{theorem} \label{thm:DSF-main}
For any constant $\epsilon > 0$, there is a polynomial-time $O(n^{3/5 + \epsilon})$-approximation algorithm for the \DSF problem with uniform edge costs.  
\end{theorem}

The current best approximation for \DSF is $\tO(n^{2/3})$~\cite{BBMRY13}, and also comes from intuition about spanners (directed spanners rather than distance preservers).  Note that while our approximation is stronger, it is for a significantly simplified setting (uniform edge costs).  

Our next result is for pairwise spanners, where we relax the exact distance requirement of preservers but unlike in \DSF do still have a hard restriction on the length.

\begin{definition}
Given a directed graph $G = (V, E)$, a collection $P \subseteq V \times V$, and an integer $k$, a subgraph $H$ of $G$ is a \emph{pairwise spanner} with multiplicative stretch $k$ if $d_H(u,v) \leq k \cdot d_G(u,v)$ for all $(u,v) \in P$.  If $d_H(u,v) \leq d_G(u,v) + k$ for all $(u,v) \in P$ then we say that $H$ has additive stretch $k$.  
\end{definition}

Our approximation algorithm will actually be for a much more general problem than finding pairwise spanners with a given stretch bound: we will allow every demand pair to have \emph{its own} stretch bound.  We can instantiate these bounds to give multiplicative or additive stretch, but we can also be more flexible if desired.  

\begin{definition}
In the \pairwise problem we are given a graph $G = (V, E)$ (possibly directed), a collection of vertex pairs $P \subseteq V \times V$, and a function $D : P \rightarrow \Nat$.  We are asked to return a subgraph $H$ minimizing $|E(H)|$ subject to $d_H(s,t) \leq D(s,t)$ for all pairs $(s,t) \in P$.  
\end{definition}

Note that this problem generalizes the other problems we have mentioned. When $D(s,t)=d_G(s,t)$ for every pair $(s,t)\in P$, this is exactly the \preserver problem. When $D(s,t)=n$ for all pairs, this is \DSF with uniform edge costs. We show the following result, which generalizes our other algorithmic results:

\begin{theorem} \label{thm:pairwise-main}
For any constant $\epsilon > 0$, there is a polynomial-time $O(n^{3/5 +\epsilon})$-approximation algorithm for the \pairwise problem.
\end{theorem}


Our algorithm must contend with the issues that come up both for \preserver and \DSF, though
the main technical obstacle is the following:
As before, we would like to find a junction tree with a common root $r$
through which we connect many terminal pairs. However, because of the
distance constraints, we cannot use the the algorithm of Chekuri et al.~\cite{CEGS11} for finding junction
trees as a black box. 
If we had constraints on the distances between
each terminal and \emph{the root $r$}, then this could be handled using a
reduction which layers the graph. However, the stretch bound for a
terminal pair here translates to a bound on the \emph{sum} of distances of
the two terminals in that pair to/from $r$.  The need to coordinate
these distance pairs across many terminal pairs, when the LP of~\cite{CEGS11}
might spread its weight across many possible distances for each
terminal, makes it more challenging to round this LP.  This requires us to go into the details of~\cite{CEGS11} and significantly change the LP and the rounding in order to allow terminals to make individual distance choices but in a way that their sums still satisfy the stretch requirements.  This involves major technical challenges and forms one of the main technical contributions of this paper.


Finally, we move to lower bounds by studying additive spanners.  

\begin{definition}
A subgraph $H$ of a (possibly directed) graph $G = (V, E)$ is a \emph{$+k$-spanner} if $d_H(u,v) \leq d_G(u,v) + k$ for all $u,v \in V$.
\end{definition}

\begin{definition}
In the \additive problem we are given a (possibly directed) graph $G = (V, E)$ and an integer $k > 0$, and are asked to return a $+k$-spanner $H$ minimizing $|E(H)|$.
\end{definition}

To the best of our knowledge, \additive has never been considered before this paper.  But while not explicit in their paper, it is straightforward to see that the $\tO(n^{1/2})$-approximation algorithm of~\cite{BBMRY13} for the multiplicative version continues to hold in the additive setting.  

\begin{theorem} \label{thm:hardness-main}
For any constant $\epsilon > 0$ and any value $k \geq 1$ (not necessarily constant), there is no polynomial-time $2^{\log^{1-\epsilon} n} / k^3$-approximation algorithm for \additive unless $NP \subseteq DTIME(2^{\polylog(n)})$.
\end{theorem}

Note that this theorem implies strong hardness results for both small and large extremes of $k$: the problem is hard when $k=1$, but is also hard when $k$ is polylogarithmic in $n$.  The associated multiplicative-stretch problem, of multiplicative $2$ and polylogarithmic in $n$ respectively, are actually both easy: there is an $O(\log n)$-approximation for the multiplicative $2$ problem~\cite{KP94}, and the classical tradeoff of~\cite{ADDJS93} implies an $O(1)$-approximation when the multiplicative stretch is at least $\Omega(\log n)$.  The $+1$-spanner hardness is particularly surprising, as many of the techniques that have been used to prove hardness for multiplicative spanners in the past~\cite{CDK12,EP00-hardness,Kor01} break down for this problem.  In particular, essentially all previous hardness results for spanners have been reductions from \textsc{Label Cover}~\cite{AL97} or \textsc{Min-Rep}~\cite{Kor01} in a way where the key factor to the hardness is the difficulty of spanning some set of crucial \emph{edges} using paths of length at most the allowed stretch.  This length must be at least $3$, since multiplicative $2$-spanners are easier to approximate~\cite{KP94}.  But for $+1$-spanners it cannot be that edges are the hard things to span, for the same reason (multiplicative $2$-spanners are easy).  Instead it must be pairs at longer distance which are hard to span.  But the obvious ways of turning an edge into a longer path (e.g.~subdividing) run into something of a catch-22, since spanning the intermediate edges of such a path is itself expensive.  So we are forced to use a much more complicated reduction which ``unifies" these paths in such a way that intermediate edges are no longer too costly, making our reduction significantly more complex.


Once we have these ideas, we can further modify them to allow us to prove hardness for much larger stretch values.  Without using the ideas from the $+1$ case it is easy to prove a theorem similar to the hardness of \basic proved in~\cite{DKR16} of $2^{(\log^{1-\epsilon} n) / k}$.  This involves starting with an instance of \textsc{Min-Rep} and subsampling the superedges to get a high girth instance, but with this technique it is not possible to prove hardness for stretch values that are logarithmic or larger.  But if we replace the superedges with longer paths, as we did in the $+1$ case, we can move the $k$ from the exponent to a multiplicative factor, allowing hardness for much larger stretch values.  This requires even more technical work than in the $+1$-case, as the longer paths introduce even more troublesome complications.

\subsection{Related Work} \label{sec:related}
Distance preservers were introduced by Coppersmith and Elkin~\cite{CE06}, who showed that extremely sparse preservers exist if the number of demand pairs is not too large.  The state of the art is due to Bodwin and Williams~\cite{BW16}, who gave a more fine-grained existential analysis and also (with~\cite{AB16-soda}) demonstrated some connections to additive spanners.  There has been much more work on additive spanners (see~\cite{Chechik16} for a recent survey); some upper bound highlights include a $+2$-spanner with $O(n^{3/2})$ edges~\cite{ACIM99}, a $+4$-spanner with $\tO(n^{7/5})$ edges~\cite{Chechik13}, and $+6$-spanner with $O(n^{4/3})$ edges~\cite{BKMP10,Woodruff10}.  From a lower bound perspective, Woodruff~\cite{Woodruff06} gave the first lower bounds that did not depend on the Erd\H{o}s girth conjecture, and recently Abboud and Bodwin~\cite{AB16-stoc} showed that the $4/3$ exponent is tight: sparser spanners are not always possible even when the allowed additive stretch is a small polynomial.

Work on optimizing spanners began over 20 years ago, with the first upper bounds provided by Kortsarz and Peleg~\cite{KP94} for stretch $2$ and the first lower bounds given by Kortsarz~\cite{Kor01} and by Elkin and Peleg~\cite{EP00-hardness} (who in particular showed hardness for pairwise spanners; hardness for distance preservers is folklore).  There has recently been a surge of progress on these problems, including new algorithms for low stretch spanners~\cite{BBMRY13,DZ16}, for directed spanners~\cite{BGJRW09,BBMRY13,DK11-stoc}, for low-degree spanners~\cite{CDK12,CD14}, and for fault-tolerant spanners~\cite{DK11-stoc,DK11-podc,DZ16}, as well as lower bounds which finally provide strong hardness~\cite{DKR16}.  

\section{An $n^{3/5+\eps}$-approximation for \preserver} \label{sec:preservers}



In this section we prove Theorem~\ref{thm:preservers-main}.  
We begin with some definitions and previous results which will help.  For any pair $(s,t) \in P$, let $\mathcal P_{s,t}$ denote the collection of shortest paths from $s$ to $t$, and let the \emph{local graph} $G^{s,t}=(V^{s,t},E^{s,t})$ be the union of all nodes and edges in $\mathcal P_{s,t}$. Note that every path from $s$ to $t$ in $G^{s,t}$ is a shortest path.  For any $k$, we say a pair $(s,t)\in P$ is $k$-\emph{thick} if $|V^{s,t}|\geq k$ and otherwise, we call the pair $k$-\emph{thin}.  We say that an edge set $F \subseteq E$ \emph{satisfies} a pair $(s,t) \in P$ if $F$ contains a shortest path from $s$ to $t$.
Consider the following standard LP relaxation (essentially that of~\cite{DK11-stoc} for spanners). 
\begin{align*}
  \min\quad&\textstyle \sum_{e\in E}x_e\\
\text{s.t.}\quad &\textstyle \sum_{p \in \mathcal P_{s,t}} f_p \geq 1&\forall{(s,t)\in P}\\
& \textstyle\sum_{p \in \mathcal P_{s,t} : e \in P} f_p \leq x_e &\forall (s,t) \in P, \forall e \in E^{s,t}\\
&x_e\geq 0&\forall e\in E
\end{align*}


Let $\mathrm{LP}$ denote the value of the optimal solution to this relaxation.  Consider the following simple randomized rounding:
\begin{itemize}
 \item Retain every edge $e\in E$ independently with probability $\min\{1,x_e\cdot k\ln n\}$.
\end{itemize}

The next claim is due to~\cite{BBMRY13} (here we adapt the wording for preservers).

\begin{claim}[\cite{BBMRY13}] \label{clm:LP}
The above rounding 
 preserves the distance for every $k$-thin pair with high probability, and the number of edges retained is at most $\tilde O(k) \cdot \mathrm{LP}$.
\end{claim}

%


For thick pairs, we have a different guarantee. As noted in~\cite{BGJRW09,DK11-stoc,BBMRY13}, the local graphs of thick pairs have a small hitting set.  This can be proved by random selection or through rounding a feasible solution for a Hitting Set LP; we include one proof here for completeness:
\begin{claim} \label{clm:hitting-set}
We can find in polynomial time a hitting set of size $\tilde O(n/k)$ for the vertex sets $V^{s,t}$ of all $k$-thick pairs $(s,t)$.
\end{claim}
\begin{proof} This is easily seen by noting that the vector $(1/k,\ldots,1/k)\in \reals^{V}$ is a feasible solution for the standard LP relaxation for Hitting Set when all sets have size at least $k$, and the value of this solution is $n/k$. Thus, it can be easily rounded to a hitting set of cardinality $\tilde O(n/k)$.
\end{proof}

We assume that we know the value $\opt$ (this is without loss of generality since we can just try every possible value in $[|E|]$ for $\opt$).  We start with the following algorithm.


\paragraph{Algorithm 1} 
\begin{itemize}
  \item Set $k=n/\sqrt{\opt}$.
  \item Run the above randomized rounding and add the edges to our current solution $F$ repeatedly until $F$ preserves distances for all $k$-thin pairs.
  \item Find a hitting set $X\subseteq V$ of cardinality $\tilde O(n/k)$ using Claim~\ref{clm:hitting-set} for the vertex sets $V^{s,t}$ for $k$-thick pairs $(s,t)$, and for every $u\in X$, add to $F$ a shortest-path tree from $u$ and a shortest-path tree into $u$.
\end{itemize}
It is easy to see that the resulting set is a preserver, and that the cardinality of the set of edges added to $F$ in the final step is $\tilde O(n^2/k)$, and so we get the following guarantee:
\begin{lemma}\label{lem:alg1} Algorithm 1 gives an $\tilde O(n/\sqrt{\opt})$-approximation for \preserver.
\end{lemma}
\begin{proof}
By the above observation for thick pairs, and Claim~\ref{clm:LP} for thin pairs, we get an approximation guarantee of $$\frac{\tilde O(k\cdot\mathrm{LP})+\tilde O(n^2/k)}{\opt}\leq \frac{\tilde O(k\cdot\mathrm{\opt})+\tilde O(n^2/k)}{\opt}=\frac{\tilde O(n\sqrt{\opt})}{\opt}=\tilde O\left(\frac{n}{\sqrt{\opt}}\right).$$
\end{proof}

\begin{remark} For global spanner problems (when $\opt=\Omega(n)$), a similar algorithm and analysis in~\cite{BBMRY13} gives an approximation ratio of $\tilde O(n/\sqrt{\opt})=\tilde O(\sqrt{n})$. However, since we have no lower bound on $\opt$ for non-global spanning and connectivity problems, we cannot achieve such a guarantee using just this algorithm.
\end{remark}

When $\opt\geq n^{4/5}$, the above guarantee immediately implies our desired $\tilde O(n^{3/5})$-approximation. So let us now consider the case when $\opt< n^{4/5}$. We begin by partitioning our pairs $P$ into $\log n$ buckets, such that in each bucket, the distances $d(s,t)$ in that bucket are all in the range $[d^*,2d^*)$ for some value $d^*$ associated with that bucket. We run our algorithm separately for each bucket, at the cost of an additional $\log n$ factor in our approximation guarantee (relative to the approximation that we will state). Let us focus on one such bucket $P_{d^*}$ for a fixed value $d^*$. Our next algorithm is similar to Algorithm 1, but does not add a full shortest-path tree and works well for buckets where $d^*$ is small.  For the sake of this algorithm, let us define sets of nearby terminals for every vertex $u\in V$: $$S^u_{d^*}=\{s\mid d(s,u)< 2d^*,\exists t:(s,t)\in P_{d^*}\}\qquad\text{and}\qquad T^u_{d^*}=\{t\mid d(u,t)< 2d^*,\exists s:(s,t)\in P_{d^*}\}$$

Our second algorithm is as follows:

\paragraph{Algorithm 2:} 
\begin{itemize}
  \item Set $k=\sqrt{d^*n}$.
  \item Run randomized rounding as in Algorithm 1 and add the edges to our current solution $F$ repeatedly until $F$ preserves distances for all $k$-thin pairs in $P_{d^*}$.
  \item Find a hitting set $X\subseteq V$ of cardinality $\tilde O(n/k)$ for the vertex sets $V^{s,t}$ for $k$-thick pairs $(s,t) \in P_{d^*}$, and for every $u\in X$, add to $F$ a shortest path from $u$ to every $t\in T^u_{d^*}$ and from every $s\in S^u_{d^*}$ to $u$.
\end{itemize}

It is easy to see that the resulting set preserves distances for all pairs in $P_{d^*}$, and that the cardinality of the set of edges added to $F$ in the final step is $\tilde O(d^* (\max_{u\in V}(|S^u_{d^*}|+|T^u_{d^*}|))n/k)$, and so we get the following guarantee:
\begin{lemma}\label{lem:alg2} 
Algorithm 2 satisfies all pairs in $P_{d^*}$ and has at most $\tilde O(\sqrt{d^*n} \cdot \opt)$ edges.
\end{lemma}
\begin{proof}
Note that every terminal must have at least one incident edge in any preserver, and so the number of terminals is always at most $O(\opt)$. Thus, by the above observation for thick pairs, and Claim~\ref{clm:LP} for thin pairs, we get a total cost of $$\tilde O(k\cdot\mathrm{LP})+\tilde O(d^*(\max_{u\in V}(|S^u_{d^*}|+|T^u_{d^*}|))n/k) \leq \tilde O(k\cdot\mathrm{\opt})+\tilde O(d^*\opt n/k)=\tilde O(\sqrt{d^*n} \cdot \opt).$$
\end{proof}

Note that for $d^*\leq n^{1/5}$, this gives us our desired $\tilde O(n^{3/5})$ approximation guarantee. Finally, we focus on the remaining case when $d^*>n^{3/5}$ (and recall that $\opt<n^{4/5}$). For this regime, we use the junction tree approach which has been used until now only for problems without hard distance constraints~\cite{CEGS11,FKN12,BBMRY13}. A junction tree is a disjoint\footnote{The two arborescences may overlap in the original graph, but we may think of them as coming from two distinct copies of the graph.} union of an in-arborescence and an out-arborescence with terminal leaves and a common root. In our case, the leaves of the in-arborescence will come from the set $S_{d^*}=\{s\mid\exists t:(s,t)\in P_{d^*}\}$ and the leaves of the out-arborescence will come from the set $T_{d^*}=\{t\mid\exists s:(s,t)\in P_{d^*}\}$. The \emph{density} of a junction tree (in the context of preservers) is the ratio between the number of edges in the junction tree, and the number of terminal pairs in $P_{d^*}$ such that the junction tree contains a shortest path connecting that pair. In order to find such a junction tree, we construct a pair of graphs for every node $u$: we let $G^{\mathrm{in}}_{d^*}(u)$ be the union of all shortest paths into $u$, let $G^{\mathrm{out}}_{d^*}(u)$ be the union of all shortest paths coming out of $u$. Note that \emph{every} path in these two graphs is a shortest path. 

Chekuri et al.~\cite{CEGS11} show that for any fixed $u$, a minimum density junction tree for general connectivity problems can be approximated within an $n^{\eps}$ factor for any $\eps>0$. While their algorithm is only designed to handle connectivity demands, and not distance-based demands, we can apply their algorithm as a black box on the graph $G_{d^*}$, which we define as the disjoint union of $G^{\mathrm{in}}_{d^*}(u)$ and $G^{\mathrm{out}}_{d^*}(u)$ (connected through $u$), and the pairs $$P_{d^*}(u)=\{(s,t)\in P_{d^*}\mid d(s,u)+d(u,t)=d(s,t)\}.$$ Since the \emph{only} paths connecting such terminal pairs in this graph are shortest paths, we are guaranteed that every pair connected by the algorithm of~\cite{CEGS11} will also be distance-preserved. Thus our third algorithm is as follows:

\noindent\textbf{Algorithm 3:} As long as $P_{d^*}\neq\emptyset$, repeat the following:
  \begin{itemize}
    \item For every $u\in V$, use the algorithm of~\cite{CEGS11} to find an approximately optimal junction tree in graph $G_{d^*}(u)$ w.r.t.\ pairs $P_{d^*}(u)$, and let $F^u$ and $P^u$ be the set of edges used and terminals connected by this algorithm, respectively.
    \item Choose $u^*$ minimizing the ratio $|F^{u^*}| / |P^{u^*}|$, add edge set $F^{u^*}$ to $F$, and remove pairs $P^{u^*}$ from $P_{d^*}$.
  \end{itemize}
The approximation guarantee of this algorithm is as follows:
\begin{lemma}\label{lem:alg3} Algorithm 3 satisfies all pairs in $P_{d^*}$ and adds at most $n^{\eps}\cdot\opt^2/d^*$ edges.
\end{lemma}
\begin{proof} As shown in~\cite{FKN12}, the above algorithm adds at most $O(\alpha \cdot \opt)$ edges (for some $\alpha>1$) as long as we have the guarantee that we can find a junction tree with density at most $\alpha\cdot\opt/|P_{d^*}|$.\footnote{They showed this in the context of Directed Steiner Forest, however the same proof applies for any terminal-pair demand problem.} Thus, by our adaptation of the algorithm of~\cite{CEGS11}, it suffices to show that there exists a junction tree with density at most $O((\opt/d^*)\cdot\opt/|P_{d^*}|)$, guaranteeing that we can find a junction tree with at most the same density up to an additional $n^\eps$ factor. To see that such a junction tree exists, let $F\subseteq E$ be some optimum solution, for every pair $(s,t)\in P_{d^*}$ fix a shortest path $p_{s,t}$ in $F$, and for every edge $e\in F$, let $P_{d^*}(e)=\{(s,t)\in P_{d^*}\mid e\in p_{s,t}\}$. Thus, for the average edge in $F$ we have
$$\frac{1}{\opt}\sum_{e\in F}|P_{d^*}(e)|=\frac{1}{\opt}\sum_{(s,t)\in P_{d^*}}|p_{s,t}|\geq\frac{1}{\opt}\sum_{(s,t)\in P_{d^*}}2d^*=\frac{2d^*|P_{d^*}|}{\opt}.$$
Thus there exists an edge $e$ with a junction tree going through it which connects at least $\Omega(d^*|P_{d^*}|/\opt)$ pairs, and uses at most $\opt$ edges (since it is a subset of $F$), giving a density of at most $O(\opt/(d^*|P_{d^*}|/\opt))=O(\opt^2/(d^*|P_{d^*}|))$, as required.
\end{proof}

So by applying Algorithm 2 for buckets where $d^* \leq n^{1/5}$ and Algorithm 3 for buckets where $d^* > n^{1/5}$, and using the fact that $\opt  < n^{4/5}$ (or else Algorithm 1 is already sufficient), Lemmas~\ref{lem:alg2} and~\ref{lem:alg3} imply that we get an approximation ratio of at most $\tilde O(\sqrt{n^{1/5} n}) + O(n^{\eps} n^{4/5} / n^{1/5}) = O(n^{3/5 + \eps})$, as claimed.



\section{Directed Steiner Forest with uniform costs} \label{sec:DSF-app}

We now turn to \DSF with uniform cost edges. 
For this problem, we would like to adapt our algorithm for preservers and get an $n^{3/5+\eps}$-approximation. This could be done quite directly if we knew, for some optimum solution $F\subseteq E$, what the distance in $F$ of every terminal pair is. However, since these distances are not known (unlike for preservers), we need to be slightly more careful. 

We start with our adaptation of Algorithm~1, for the corresponding case when $\opt\geq n^{4/5}$, and we use the following standard LP relaxation:
\begin{align*}
  \min\quad&\textstyle\sum_{e\in E}x_e\\
\text{s.t.}\quad &\text{capacities }x_e\text{ support a one unit }s-t\text{ flow }f_{s,t}&\forall{(s,t)\in P}\\
&x_e\geq 0&\forall e\in E
\end{align*}


Here we have no notion of local graphs independent of the LP. Instead, for every pair $(s,t)\in P$, we define its local graph w.r.t.\ this LP: let $V^{s,t}$ be the set of vertices involved in the flow $f_{s,t}$. We can define $k$-thick and $k$-thin terminal pairs accordingly -- $(s,t)$ is a $k$-thick pair if $|V^{s,t}| \geq k$, and otherwise it is $k$-thin. The guarantee of Lemma~\ref{lem:alg1} follows here with the same analysis. The only difference is the justification of Claim~\ref{clm:LP} from~\cite{BBMRY13}. In their paper, the proof of this claim relies entirely on showing that the number of minimal ``antispanners" (edge sets whose removal from the flow $f_{s,t}$ increases the distance from $s$ to $t$) is bounded by $|E|\cdot|V^{s,t}|^{|V^{s,t}|}$. In our case, the proof is even simpler, since rather bounding the number of antispanners, we need to bound the number of minimal cuts, which is clearly at most $2^{|V^{s,t}|}$ (since every minimal cut is simply a partition).


Now suppose $\opt<n^{4/5}$. As discussed in the proof of Lemma~\ref{lem:alg3}, as long as we have an algorithm to find a junction tree (or really any edge set) with density at most $\alpha\cdot\opt/|P|$, we can repeatedly apply such an algorithm to the set of remaining unconnected pairs, adding at most $O(\alpha\cdot\opt)$ edges overall. One way to find an edge set with good density is indeed via junction trees, or by adapting one of our other algorithms to connect a constant fraction of pairs. Thus, we need an algorithm which finds an edge set with density at most $n^{3/5+\eps}\cdot\opt/|P|$. Rather than using buckets as in Section~\ref{sec:preservers}, we use thresholding. While our approach at this point is very similar to that of~\cite{BBMRY13}, the actual parameters and combination of algorithmic components is slightly different, so we give the algorithm and analysis here for completeness. Set $D=n^{1/5}$, and define a layered graph $G_D=(V_D,E_D)$ where $V_D=V\times\{0,1,\ldots, D\}$, and $$E_D=\{((u,j-1),(v,j))\mid j\in[D]\wedge(((u,v)\in E)\vee (u=v\in V))\}$$ We start with the following LP relaxation, where $|f_{s,t}|$ denotes the value of a flow $f_{s,t}$:
\begin{align}
  \min\quad&\textstyle\sum_{e\in E}x_e\nonumber\\
\text{s.t.}\quad &f_{s,t}\text{ is a }(s,0)-(t,D)\text{ flow}&\forall{(s,t)\in P}\\
&\sum_{j=1}^{D}f_{s,t}((u,j-1),(v,j))\leq x_{(u,v)}&\forall (u,v)\in E,(s,t)\in P\\
&|f_{s,t}|\leq 1&\forall (s,t)\in P\label{LP:unit-flow}\\
&\sum_{(s,t)\in P}|f_{s,t}|\geq |P|/2\label{LP:lotsa-flow}\\
&x_e\geq 0&\forall e\in E
\end{align}
While this LP seems to increase the number of vertices by a factor of $D$, it is simply a compact way of formulating a flow LP on $G$ where the flows are restricted to paths of length at most $D$, and we may think of the flows $f_{s,t}$ as such.

The algorithm now trades off a slight variant of Algorithm 2 and Algorithm 3, and chooses the sparser of the two edge sets:

\noindent\textbf{Algorithm 4:}
\begin{itemize}
  \item Solve the above LP relaxation (if feasible), and run Algorithm 2 on pairs $P_{1/4}:=\{(s,t)\in P\mid |f_{s,t}|\geq 1/4\}$, where $V^{s,t}$ is the support (in $V$) of flow $f_{s,t}$, and distance bounds $D$ (w.r.t.\ the hitting set), and let $F_0$ be the edge set found by this algorithm.
  \item For every $u\in V$, use the algorithm of~\cite{CEGS11} to find an approximately optimal junction tree in graph $G$ and all pairs $P$, and let $F_u$ and $P_u$ be the set of edges used and terminals connected by this algorithm, respectively.
  \item If the LP is infeasible, or $|F_0|/|P|>\min_u |F_u|/|P_u|$, then add output $F_u$ that minimizes this ratio. Otherwise, output $F_0$.
\end{itemize}

The following lemma gives the required guarantee:
\begin{lemma}\label{lem:whatever} Assuming $\opt<n^{4/5}$, Algorithm 4 outputs an edge set $F'$ which connects terminal pairs $P'\subseteq P$ such that $|F'|/|P'|\leq n^{3/5+\eps}\cdot\opt/|P|$.
\end{lemma}
\begin{proof}
For some optimum solution, let $P_D$ be the set of terminal pairs in $P$ which the solution connects using paths of length at most $D$. Consider two cases.\\

\noindent\textbf{Case 1}: $|P_D|\geq|P|/2$.\\
In this case, the above LP is feasible, and has value at most $\opt$. Also note that constraints~\eqref{LP:unit-flow} and~\eqref{LP:lotsa-flow} imply that $|P_{1/4}|\geq|P|/4$, and so by the same analysis as Lemma~\ref{lem:alg2}, the algorithm connects these pairs using at most $\tilde O(n^{3/5}\opt)$ edges, giving a set of density $\tilde O(n^{3/5}\cdot\opt/|P|)$.\\

\noindent\textbf{Case 2}: $|P_D|<|P|/2$.\\
In this case, at least $|P|/2$ pairs are connected by paths of length at least $D$, and so we find a junction tree with density at most $n^{\eps}\cdot\opt^2/(|D||P|)<n^{3/5+\eps}\cdot\opt/|P|$ by the exact same analysis as in the proof of Lemma~\ref{lem:alg3}.
\end{proof}


\section{Pairwise spanners} \label{sec:pairwise}


Much of our algorithm for \preserver and/or \DSF  can be used for \pairwise with almost no change. As before, we guess $\opt$, and run Algorithm~1 if $\opt\geq n^{4/5}$, since the guarantee of Lemma~\ref{lem:alg1} applies here as well.

If we were concerned only with constant additive or multiplicative stretch spanners, we could bucket on distances as in our algorithm for \preserver. However, in \pairwise we could have arbitrary distance demands, so this bucketing is not possible in general. Instead, we use the same approach as Algorithm 4 for \DSF, with the same analysis as in Lemma~\ref{lem:whatever} plus two changes. First, we need to very slightly change our application of Algorithm~2. As in Section~\ref{sec:DSF-app}, we define the graph $G_{D_0}=(V_{D_0},E_{D_0})$ for $D_0=n^{1/5}$, and solve the following slightly modified LP:
\begin{align}
  \min\quad&\sum_{e\in E}x_e\nonumber\\
\text{s.t.}\quad &f_{s,t}\text{ is a }(s,0)-(t,\min\{D_0,D(s,t)\})\text{ flow}&\forall{(s,t)\in P}\\
&\sum_{j=1}^{D_0}f_{s,t}((u,j-1),(v,j))\leq x_{(u,v)}&\forall (u,v)\in E,(s,t)\in P\\
&|f_{s,t}|\leq 1&\forall (s,t)\in P\\
&\sum_{(s,t)\in P}|f_{s,t}|\geq |P|/2\\
&x_e\geq 0&\forall e\in E
\end{align}
As shown in~\cite{BBMRY13}, this LP ensures that thin pairs $(s,t)$ that are settled by Algorithm 2 will be spanned by a path of length at most $D(s,t)$ (the general spanner version of Claim~\ref{clm:LP}).

However, the step in Algorithm~4 that requires junction trees cannot work as stated. Nor can we apply known techniques via a simple reduction as in Algorithm~3. Recall that in Algorithm 3, we constructed a graph of shortest paths through a root $r$ (which we called $u$ in Section~\ref{sec:preservers}) for which we wanted to find a sparsest junction tree using the algorithm of~\cite{CEGS11} as a black box. This black box reduction worked because we restricted our instance so that all paths were shortest paths. However, there is no obvious graph we can restrict to for arbitrary distance bounds (or even when we only allow additive or multiplicative stretch). Consider a single pair $(s,t)$, and all paths of distance $D(s,t)$ that go through a proposed junction tree rooted at $r$. Some of these paths could have a prefix of length $D(s,t)/2$ to $r$ and a suffix of length $D(s,t)/2$ from $r$ to $t$, while others may have a prefix of length $D(s,t)/2-1$ to $r$ and a suffix of length $D(s,t)/2+1$ from $r$ to $t$. Allowing the algorithm to connect $s$ and $t$ using the edges from the union of such paths can result in a path of length $D(s,t)+1$, or even more, violating the required distance.

Thus handling these hard distance requirements introduces significant difficulties, and requires us to modify the algorithm of~\cite{CEGS11} in a non-black-box manner. At a high level, their algorithm uses height reduction and a reduction to (undirected) \textsc{Group Steiner Tree} in a tree via a density LP. Roughly speaking, our version of their density LP has a flow $f_s$ to the root $r$ for every left terminal $s$, and a flow $f_t$ from $r$ to $t$ for every right terminal $t$. Each flow is explicitly decomposed into flows on specific path lengths. That is, $f^k_s$ only uses paths of length exactly $k$, etc. Our LP is roughly of the form\footnote{In reality we apply a  similar LP relaxation to a graph which has undergone a series of transformations, including layering and height reduction.}
\begin{align*}
  \min\quad&\sum_{e\in E}x_e\nonumber\\
\text{s.t.}\quad &\text{capacities }x_e\text{ support }s-r,r-t\text{ flows }f_s,f_t&\forall\text{ terminals }s,t\\
&\sum_{(s,t)\in P}\;\sum_{k+\ell\leq D(s,t)}y_{s,t}^{k,\ell}=1\\
&\sum_{\ell=0}^{D(s,t)-k}y_{s,t}^{k,\ell}\leq |f^k_s|&\forall(s,t)\in P\\
&\sum_{k=0}^{D(s,t)-\ell}y_{s,t}^{k,\ell}\leq |f^\ell_t|&\forall(s,t)\in P\\
&x_e\geq 0&\forall e\in E\\
&y_{s,t}^{k,\ell}\geq 0&\forall (s,t)\in P,k,\ell\in[n-1]
\end{align*}

In order to apply the technique of~\cite{CEGS11}, we need to make the choices of distances for left and right terminals independent. In other words, we need to restrict every left or right terminal $u$ to a fixed set of distances $\Lambda(u)$ such that for every pair $(s,t)\in P$ we have $k+\ell\leq D(s,t)$ for \emph{every} pair of distances $(k,\ell)\in \Lambda(s)\times\Lambda(t)$, while still preserving a large ($1/\polylog(n)$) fraction of the total flow. It turns out that this can be achieved by retaining, for every terminal, all distances up to its median distance, weighted by the $y$ variables. Since this part of our algorithm is somewhat more complex, and involves a non-blackbox use of previous junction tree algorithms, we present our junction tree algorithm separately in Section~\ref{sec:junction-tree-pairwise-spanner}. 



\section{Junction Tree Algorithm for \pairwise}
\label{sec:junction-tree-pairwise-spanner}

In this section, we present our approximation algorithm
for approximating the Minimum Density Junction Tree for 
the \pairwise problem, which we then use as part of Algorithm~3 in
Section~\ref{sec:pairwise} (replacing Algorithm 3 from Section~\ref{sec:preservers}). Formally, we prove the following theorem:

\begin{theorem}\label{thm:pairwise-spanner-junction-tree} For any constant $\eps>0$, there is a polynomial time algorithm which, given an unweighted directed $n$ vertex graph $G=(V,E)$, terminal pairs $P\subseteq V\times V$, and distance bounds $D:P\rightarrow\Nat$ (where $D(s,t)\geq d(s,t)$ for every terminal pair $(s,t)\in P$, approximates the following problem to within an $O(n^{\eps})$ factor:
\begin{itemize}
   \item Find a non-empty set of edges $F\subseteq E$ minimizing the ratio $$\min_{r\in V}\frac{|F|}{|\{(s,t)\in P\mid d_{F,r}(s,t)\leq D(s,t)\}|},$$ where $d_{F,r}(s,t)$ is the length of the shortest path using edges in $F$ which connects $s$ to $t$ while going through $r$ (if such a path exists).
\end{itemize}
\end{theorem}

Before we describe our algorithm and analysis, let us first mention two tools on which our algorithm relies.

\subsection{Useful lemmas from previous work}

The following lemma describes part of junction tree algorithm of~\cite{CEGS11}. Since it is not explicitly mentioned as a separate result in their paper, we give a proof sketch in Appendix~\ref{sec:dist-height-reduction} for completeness.
\begin{lemma}[Height reduction]
\label{lem:height-reduction}
Let $G=(V,E)$ be an edge-weighted directed graph with edge weights $w:E\rightarrow\reals^{\geq 0}$, let $r\in V$ be a source vertex of $G$,
and let $\sigma>0$ be some parameter.
Then we can efficiently construct an edge-weighted undirected tree $\hat{T}_r$ rooted at  
$\hat{r}$ of height $\sigma$ and size $|V|^{O(\sigma)}$ together with 
edge weights $\hat w:E(\hat{T}_r)\rightarrow \mathbb{R}^+_0$, and a vertex mapping $\Psi:V(\hat{T}_r)\rightarrow V(G)$, such that
\begin{itemize}
\item For any arborescence $J\subseteq G$ rooted at $r$, and terminal set $S\subseteq J$,
      there exists a tree $\hat{J}\subseteq \hat{T}_r$ rooted at $\hat r$ such that letting $\mathcal{L}(J)$ and $\mathcal{L}(\hat J)$ be the set of leaves of $J$ and $\hat J$, respectively, we have
      $\mathcal L(J) = \Psi(\mathcal L(\hat{J}))$. 
      Moreover, $w(J) \leq 
      O(\sigma|\mathcal{L}|^{1/\sigma})\cdot\hat w(\hat{J})$.
\item Given any tree $\hat{J}\subseteq \hat{T}_r$ rooted at $\hat{r}$,
      we can efficiently find an arborescence $J\subseteq G$ rooted at~$r$ such that, for leaf sets $\mathcal{L}(J)$ and $\mathcal{L}(\hat J)$ as above, we have $\mathcal L(J) = \Psi(\mathcal L(\hat{J}))$, and
      moreover, $w(J) \leq \hat w(\hat{J})$.
\end{itemize}
\end{lemma}

\noindent The following lemma is from~\cite{GKR00} (Theorem 4.1, Corollary 6.1), and was also used in~\cite{CEGS11}.  Given a set $X$, we will use $\mathcal P(X)$ and $2^X$ interchangeably to denote the power set of $X$.  
\begin{lemma}[Group Steiner Tree rounding in a tree] \label{lem:GKR} Given an edge-weighted undirected tree~$T$ rooted at $r$, with edge weights $w:E(T)\rightarrow \reals^{\geq 0}$, a collection of vertex sets $\mathcal{S}\subseteq\mathcal P(V(T))$, 
  and a solution to the following LP:
\begin{align*}
\min\quad & \displaystyle\sum_{e\in E}w(e)x_e\qquad&\\
\text{s.t.}\quad
&\text{capacities $\{x_e\}$ support one unit of flow 
       from $r$ to $S$}
   &\forall S\in\mathcal{S}\\
& x_e\geq 0
   &\forall e\in E
\end{align*}
Then we can efficiently find a subtree $T'\subseteq T$ rooted at $r$ such that for every $S\in\mathcal S$ at least one vertex of $S$ participates in the tree $T'$, and $w(T')\leq O(\log n\log |\mathcal S|\cdot \sum_{e\in E}w(e)x_e)$.
\end{lemma}

\subsection{Our junction tree algorithm}

Although we solve the same Minimum Density Junction Tree problem
as in \preserver and \DSF, there are key technical difficulties we run into when trying to find junction trees for \pairwise.
In \DSF, we do not have to concern ourselves with the length of the paths.
In \preserver, we are able to restrict the graph to use only
shortest paths. 
Here we have many choices of source-to-root and root-to-sink
paths with various distances,
and we have to pick a pair of paths with the right distances.
Thus, it is not possible to use the algorithm of~\cite{CEGS11}
for the Minimum Density Junction Tree problem as a black box. While our
algorithm follows the same overall structure as theirs, key parts of the algorithm
need to be applied in conjunction with a reduction, while others need to be
replaced.

In fact, we have two technical issues.
First, a standard technique that reduces 
the Minimum Density Junction Tree problem to
a tree instance of the density version of the \textsc{Group Steiner Forest} problem (GSF)
does not work for us, as it does not keep track of distances. Secondly,
the approach of~\cite{CEGS11} is to bucket an LP relaxation for Density GSF, giving a relaxation for
GSF itself, which can then be rounded directly using the algorithm of~\cite{GKR00}. Here, again,
we cannot apply their technique directly, because the bucketing needs to prune the possible
distances for different terminals, so that they can be picked independently while still respecting
the total distance bound between terminal pairs.

As usual, we will try all possibilities for a root $r\in V$, and approximate a minimum density junction tree rooted at $r$. For every such choice of $r$, we start by reducing our problem to an instance of the following problem:
\begin{definition}
In the \textsc{Minimum Density Steiner Label Cover} problem we are given a directed graph $G = (V, E)$, nonnegative edge costs $w: E \rightarrow \reals^{\geq 0}$, two collections of disjoint vertex sets $\mathcal S,\mathcal T\subseteq \mathcal P(V)$, a collection of set pairs $P \subseteq \mathcal S \times \mathcal T$, and for each pair $(S,T)\in P$, a relation $R(S,T)\subseteq S\times T$. The goal is to find a set of edges $F\subseteq E$ minimizing the ratio $$\frac{w(F)}{|\{(S,T)\in P\mid \exists(s,t)\in R(S,T):F\text{ contains an }s\leadsto t\text{ path}\}|}.$$
\end{definition}

Our reduction allows us to use the first part of~\cite{CEGS11} by turning our distance problem into a connectivity problem, albeit a considerably more complicated one. The reduction is as follows: construct a layered directed graph with vertices $$V_r=((V\setminus r)\times\{-n+1,\ldots,-2,-1,1,2,\ldots,n-1\})\cup\{(r,0)\}$$ and edges $$E_r=\{((u,i),(v,i+1))\mid (u,i),(v,i+1)\in V_r,(u,v)\in E\}.$$ Set all edge weights for $e\in E_r$ to $w(e)=1$. Now for every terminal pair $(s,t)\in P$ with distance bound $D(s,t)$, add new vertices $(s^t,-i)$ and $(t^s,j)$ for all $i,j\geq 0$ such that $(s,-i),(t,j)\in V_r$, and for all such $i$ and $j$ add zero-weight edges $((s^t,-i),(s,-i))$ and $((t,j),(t^s,j))$. Denote this final graph by $G_r$. Finally, for every terminal pair $(s,t)$, define terminal sets $S_{s,t}=\{(s^t,-i)|i\geq 0\}\cap V(G_r)$ and $T_{s,t}=\{(t^s,j)|j\geq 0\}\cap V(G_r)$, and relation $R_{s,t}=\{((s^t,-i),(t^s,j))\in S_{s,t}\times T_{s,t}\mid i+j\leq D(s,t)\}$.

Note that for every terminal pair $(s,t)\in P$ and label $i$, there is a bijection between paths of length $i$ from $s$ to $r$ in $G$, and paths from $(s^t,-i)$ to $(r,0)$ in $G_r$, and similarly a bijection between paths of length $i$ from $r$ to $t$ in $G$, and paths from $(r,0)$ to $(t^s,i)$. This reduces the problem of ``keeping track" of path lengths in $G$ to connecting appropriately labeled terminal pairs in $G_r$. It also creates disjoint terminal pairs, by creating a separate copy $s^t$ of $s$ for every terminal pair $(s,t)\in P$ that $s$ participates in (and similarly for terminals $t$). This will simplify our later algorithm and analysis.

To use our construction, we need one more, much simpler graph: Let $G'$ be a graph comprised of two graphs $G_+$ and $G_-$ which are copies of $G$ intersecting only in vertex $r$. For every node $u\in V$, denote by $u_+,u_-$ the copies of $u$ in $G_+,G_-$, respectively. The following lemma follows directly from our construction:
\begin{lemma}\label{lem:layered}  For any $f>0$, and set of terminal pairs $P'\subseteq P$, there exists an edge set $F\subseteq E(G')$ of size $|F|\leq f$ containing a path of length $\leq D(s,t)$ from $s_-$ to $t_+$ for every $(s,t)\in P'$ iff there exists a junction tree $J\subseteq E(G_r)$ of weight $w(J)\leq f$ such that for every terminal pair $(s,t)\in P'$, $J$ contains leaves $(s^t,-i),(t^s,j)$ such that $((s^t,-i),(t^s,j))\in R_{s,t}$. Moreover, given such a junction tree $J$, we can efficiently find a corresponding edge set $F$.
\end{lemma}

Thus, to prove Theorem~\ref{thm:pairwise-spanner-junction-tree}, it suffices to show that we can achieve an $O(n^{\eps})$-approximation for the \textsc{Minimum Density Steiner Label Cover} instance $(G_r,w,\{S_{s,t},T_{s,t},R_{s,t}\mid (s,t)\in P\})$ obtained from our reduction. We are now ready to apply the first part of the junction tree algorithm of~\cite{CEGS11}. We apply the algorithm of Lemma~\ref{lem:height-reduction} to our weighted graph $(G_r,w)$ with parameter (constant) $\sigma>1/\eps$, and obtain a shallow tree $\hat T_r$ with weights $\hat w$, and mapping $\Psi:V(\hat T_r)\rightarrow V(G_r)$. All terminals $(s^t,-i),(t^s,j)$ in $G_r$ are now represented by sets of terminals $\Psi^{-1}((s^t,-i))$ and $\Psi^{-1}(t^s,j)$, respectively, in $V(\hat T_r)$.\footnote
{Note that $\Psi$ is, in fact, not a bijection. We abuse $\Psi^{-1}(v)$ to mean $\{\hat{u} \in V(\hat T_r) \mid \Psi(\hat{u}) = v\}$.}
We can extend the relation $R_{s,t}$ in the natural way to $$\hat R_{s,t}=\{(\hat s,\hat t)\in \Psi^{-1}(S_{s,t})\times\Psi^{-1}(T_{s,t})\mid(\Psi(\hat s),\Psi(\hat t))\in R_{s,t}\}.$$ Therefore, by Lemma~\ref{lem:layered} and Lemma~\ref{lem:height-reduction}, it suffices to show the following:

\begin{lemma}\label{lem:LP-rounding} There exists a polynomial time algorithm which, in the above setting, gives an $O(\log^3 n)$ approximation for the following problem:
\begin{itemize}
  \item Find a tree $T\subseteq\hat T_r$ minimizing the ratio $$\frac{\hat w(E(T))}{|\{(s,t)\in P\mid \exists (\hat s,\hat t)\in \hat R_{s,t}:T\text{ contains an }\hat s\leadsto\hat t\text{ path}\}|}.$$
\end{itemize}
\end{lemma}

The rest of this section is devoted to proving Lemma~\ref{lem:LP-rounding}. In the original junction tree algorithm of~\cite{CEGS11}, with no distance constraints, the application of Lemma~\ref{lem:height-reduction} gave a \textsc{Minimum Density Group Steiner Forest} instance, for which they formulated and rounded an LP relaxation. However, our reduction yields a connectivity problem of a much more subtle nature, in which it is not enough to choose representatives of various sets of terminals, but to also choose them in a way that satisfies the relations $\hat R(s,t)$. Fortunately, the relations $\hat R(s,t)$ have a very specific structure, which still allows us to obtain a polylogarithmic approximation (whereas for more general relations the problem would be considerably harder). Nevertheless, these non-cross-product relations require both a new LP relaxation, and a more subtle rounding. We use the following natural LP relaxation for the problem described in Lemma~\ref{lem:LP-rounding}:

\begin{align}
\min\quad & \displaystyle\sum_{e\in E(\hat T_r)}\hat w(e)x_e\qquad&\nonumber\\
\text{s.t.}\quad
   &\sum_{(s,t)\in P}\sum_{(\hat s,\hat t)\in \hat R_{s,t}}y_{\hat s,\hat t}=1\label{LP:normalize}\\
   &\sum_{\hat t:(\hat s,\hat t)\in R_{s,t}}y_{\hat s,\hat t}\leq z_{\hat s}
      &\forall (s,t)\in P,\hat s\in\Psi^{-1}(S_{s,t})\label{LP:z_s-bound}\\
   &\sum_{\hat s:(\hat s,\hat t)\in R_{s,t}}y_{\hat s,\hat t}\leq z_{\hat t}
      &\forall (s,t)\in P,\hat t\in\Psi^{-1}(T_{s,t})\label{LP:z_t-bound}\\
&\text{capacities $x_e$ support $z_{\hat s}$ flow 
       from $\hat s$ to $\hat{r}$}
   &\forall (s,t)\in P,\hat s\in\Psi^{-1}(S_{s,t})\label{LP:s-flow}\\
&\text{capacities $x_e$ support $z_{\hat t}$ flow
       from $\hat{r}$ to $\hat t$}
      &\forall (s,t)\in P,\hat t\in\Psi^{-1}(T_{s,t})\label{LP:t-flow}\\
 &y_{\hat s,\hat t}\geq 0
     &\forall (s,t)\in P,(\hat s,\hat t)\in\hat R_{s,t}\\
 &x_e\geq 0
     &\forall e\in E(\hat T_r)
\end{align}
This is easily seen to be a relaxation, by considering, for any tree $T\subseteq \hat T_r$, the following solution: Let $P_T=\{(s,t)\in P\mid \exists (\hat s,\hat t)\in \hat R_{s,t}:T\text{ contains an }\hat s\leadsto\hat t\text{ path}\}$.\footnote{We describe the paths and LP flows as directed for convenience, but recall that the tree $\hat T_r$ is undirected, so there is no particular meaning to the direction of the flows and paths other than the implicit connection to the original, directed graph.} Let $x_e=1/|P_T|$ for every $e\in E(T)$, and $x_e=0$ otherwise. Then clearly the objective function gives $\sum_{e\in E(\hat T_r)}\hat w(e)x_e=\hat w(T)/|P_T|$ as required. Next, for every $(s,t)\in P_T$, let $(\hat s,\hat t)\in\hat R_{s,t}$ be a pair of representatives that is connected by $T$, and set $z_{\hat s}=z_{\hat t}=y_{\hat s,\hat t}=1/|P_T|$, and set all other $y$ and $z$ variables to $0$. It is easy to check that all constraints are satisfied by this solution.

Given an optimum solution to the above LP relaxation, our goal is to transform such a solution to an LP relaxation for \textsc{Group Steiner Tree} (on the tree $\hat T_r$), which can then be rounded using known algorithms without losing too much in the LP value. In particular, for pairs $(s,t)\in P$, we need to prune the sets of representatives $\Psi^{-1}(S_{s,t})$ and $\Psi^{-1}(T_{s,t})$, so that we can choose such representatives independently while still respecting the problem structure. Formally, we need to find representative sets $\tilde S_{s,t}\subseteq\Psi^{-1}(S_{s,t})$ and $\tilde T_{s,t}\subseteq\Psi^{-1}(T_{s,t})$ such that $\tilde S_{s,t}\times\tilde T_{s,t}\subseteq \hat R_{s,t}$ (so that there is no chance that we will pick a pair of representatives corresponding to a path of total length greater than $D(s,t)$), but also make sure that these representative sets still cover a large LP value. This representative set pruning is accomplished by sorting the representatives of each terminal, and taking all representatives of a terminal up to the median representative. The formal pruning procedure is as follows:
\begin{itemize}
  \item For all terminal pairs $(s,t)\in P$, define $\gamma_{s,t}:=\sum_{(\hat s,\hat t)\in \hat R_{s,t}}y_{\hat s,\hat t}$.
  \item For every terminal pair $(s,t)\in P$, sort the representative sets $\Psi^{-1}(S_{s,t})$ and $\Psi^{-1}(T_{s,t})$ by non-decreasing order of distance labels:
  \begin{itemize}
    \item Sort $\Psi^{-1}(S_{s,t})=\{\hat s_1,\hat s_2,\ldots\}$ so that for any $i',j'$, if $\Psi(\hat s_{i'})=(s^t,-i)$ and $\Psi(\hat s_{j'})=(s^t,-j)$ for some $i<j$, then $i'<j'$.
    \item Sort $\Psi^{-1}(T_{s,t})=\{\hat t_1,\hat t_2,\ldots\}$ so that for any $i',j'$, if $\Psi(\hat t_{i'})=(t^s,i)$ and $\Psi(\hat t_{j'})=(t^s,j)$ for some $i<j$, then $i'<j'$.
  \end{itemize}
  \item Choose median prefix sets for all terminals: For every terminal pair $(s,t)\in P$, define
$$\mu(s^t):=\min\left\{k\;\left|\; \sum_{i=1}^{k}\sum_{\hat t:(\hat s_i,\hat t)\in \hat R_{s,t}}y_{\hat s_i,\hat t}\geq\gamma_{s,t}/2\right.\right\},\text{ and}$$
$$\mu(t^s):=\min\left\{k\;\left|\; \sum_{i=1}^{k}\sum_{\hat s:(\hat s,\hat t_i)\in \hat R_{s,t}}y_{\hat s,\hat t_i}\geq\gamma_{s,t}/2\right.\right\},$$
and let $$\tilde S_{s,t}:=\{\hat s_i\mid i\in [\mu(s^t)]\}\qquad\text{and}\qquad \tilde T_{s,t}:=\{\hat t_i\mid i\in [\mu(t^s)]\}.$$
\end{itemize}
Note that our pruning is somewhat simplified by the fact that the terminal representative sets $\Psi^{-1}(S_{s,t})$ and $\Psi^{-1}(T_{s,t})$ are disjoint from the representative sets for all other pairs, so we do not need to prune the representatives of a single terminal $s$ w.r.t.\ a number of terminals $t$ simultaneously (or vice versa). As we shall see, the choice of median representative prefix sets automatically guarantees that at least half the LP value is preserved. However, we need to verify that these sets do not contain any pairs that are disallowed by the spanner constraints (formally, by the terminal relations $\hat R_{s,t}$):
\begin{lemma}\label{lem:spanning-pairs} In the above algorithm, for every terminal pair $(s,t)\in P$ we have $\tilde S_{s,t}\times \tilde T_{s,t}\subseteq \hat R_{s,t}$.
\end{lemma}
\begin{proof}
Fix some terminal pair $(s,t)\in P$. For the purpose of this proof, let us also define the suffix set $U=\{\hat s_i\mid i\in\{\mu(s^t),\ldots,|\Psi^{-1}(S_{s,t})|\}\}$. By our choice of $\mu(s^t)$ and definition of $\gamma_{s,t}$, we have that
\[
\sum_{\hat s\in U}\sum_{\hat t:(\hat s,\hat t)\in \hat R_{s,t}}y_{\hat s,\hat t}
=\gamma_{s,t}-\sum_{i=1}^{\mu(s^t)-1}\sum_{\hat t:(\hat s_i,\hat t)\in \hat R_{s,t}}y_{\hat s_i,\hat t}
> \gamma_{s,t}-\gamma_{s,t}/2
= \gamma_{s,t}/2.
\]
Now, let $\ell$ be such that $\Psi(\hat s_{\mu(s^t)})=(s^t,-\ell)$, and note by our sorted ordering that for every $\hat s\in U$ we have $\Psi(\hat s)=(s_t,-\ell')$ for some $\ell'\geq\ell$. Therefore, for any representative $\hat s\in U$, for any representative $\hat t\in T_{s,t}$ such that $(\hat s,\hat t)\in \hat R_{s,t}$ we must have $\Psi(\hat t)=(t^s,j)$ for some $j\leq D(s,t)-\ell'\leq D(s,t)-\ell$. Thus, letting $k=\max\{k'\mid \Psi(\hat t_{k'})=(t^s,D(s,t)-\ell)\}$, again by our sorted ordering, we have that for any $\hat s\in U$ and $k'\geq 0$,  if $(\hat s,\hat t_{k'})\in\hat R_{s,t}$ then $k'\in[k]$. This gives
$$\sum_{i=1}^{k}\sum_{\hat s:(\hat s,\hat t_i)\in \hat R_{s,t}}y_{\hat s,\hat t_i}\geq \sum_{i=1}^{k}\sum_{\substack{\hat s\in U\\(\hat s,\hat t_i)\in\hat R_{s,t}}}y_{\hat s,\hat t_i}=\sum_{\hat s\in U}\sum_{\hat t:(\hat s,\hat t)\in \hat R_{s,t}}y_{\hat s,\hat t}\geq\gamma/2.$$
This immediately implies that $k\geq\mu(t^s)$. Moreover, for any $\hat s\in \tilde S_{s,t}$ and $\hat t\in \tilde T_{s,t}(\subseteq\{\hat t_1,\ldots,\hat t_k\})$, we have $\Psi(\hat s)=(s^t,-\ell')$ and $\Psi(\hat t)=(s^t,j')$ such that $\ell'\leq \ell$ and $j'\leq k=D(s,t)-\ell$, and therefore $\ell'+j'\leq D(s,t)$, which implies that $(\hat s,\hat t)\in\hat R_{s,t}$, thus proving the lemma.
\end{proof}

With this pruning in place, we can run the remaining part of the classical junction tree algorithm. By the definition of $\gamma_{s,t}$ and constraint~\eqref{LP:normalize}, we get that $\sum_{(s,t)\in P}\gamma_{s,t}=1$. We now bucket the pairs $P$ by their $\gamma$ values. That is, for all $i\in\{0,1,\ldots,\lceil\log|P|\rceil\}$, define $$P_i=\{(s,t)\in P\mid \gamma_{s,t}\in(2^{-i-1},2^{-i}]\}.$$ By a standard argument, there exists some $i^*$ such that $\sum_{(s,t)\in P_{i^*}}\gamma_{s,t}\geq\frac 1{2(\lceil\log|P|\rceil+1)}$, and so $|P_{i^*}|\geq 2^{i^*}/O(\log n)$. Moreover, for every pair $(s,t)\in P_{i^*}$ we have
\begin{align*}
\sum_{\hat s\in\tilde S_{s,t}}z_{\hat s}&\geq \sum_{\hat s\in\tilde S_{s,t}}\sum_{\hat t: (\hat s,\hat t)\in\hat R_{s,t}}y_{\hat s,\hat t}&\text{by constraint~\eqref{LP:z_s-bound}}\\
&\geq\gamma_{s,t}/2\geq 2^{-i^*-2},
\end{align*}
and similarly $$\sum_{\hat t\in\tilde T_{s,t}}z_{\hat t}\geq 2^{-i^*-2}.$$

Now, scaling our LP solution up to $x^*_e:=\min\{1,2^{i^*+2}\cdot x_e\}$, from the above bounds together with constraints~\eqref{LP:s-flow} and~\eqref{LP:t-flow}, we get a (possibly suboptimal) solution to the following LP:

\begin{align*}
\min\quad & \displaystyle\sum_{e\in E}w(e)x^*_e\qquad&\\
\text{s.t.}\quad
&\text{capacities $\{x^*_e\}$ support one unit of flow 
       from $\hat r$ to $\tilde S_{s,t}$}
   &\forall (s,t)\in P_{i^*}\\
&\text{capacities $\{x^*_e\}$ support one unit of flow 
       from $\hat r$ to $\tilde T_{s,t}$}
   &\forall (s,t)\in P_{i^*}\\
& x^*_e\geq 0
   &\forall e\in E
\end{align*}

By Lemma~\ref{lem:GKR}, we can round this LP solution, and obtain a tree $T'\subseteq \hat T_r$ of weight $$w(T')=O\left(\log^2n\cdot \sum_{e\in E(\hat T_r)}w(e)x^*_e\right)\leq 2^{i^*}\cdot O\left(\log^2n\cdot \sum_{e\in E(\hat T_r)}w(e)x_e)\right)$$ such that for every pair $(s,t)\in P_{i^*}$ at least one vertex $\hat s\in\tilde S_{s,t}$ and at least one vertex $\hat t\in\tilde T_{s,t}$ are connected through $\hat r$ in the tree $T'$. Recalling by Lemma~\ref{lem:spanning-pairs} that such $(\hat s,\hat t)$ pairs also belong to $\hat R_{s,t}$, we have the following bound on the ``density" of the tree $T'$:
$$\frac{w(T')}{|\{(s,t)\in P\mid\exists(\hat s,\hat t)\in\hat R_{s,t}:T'\text{ contains an }\hat s\leadsto\hat t\text{ path}\}|}\leq\frac{w(T')}{|P_{i^*}|}=O\left(\log^3 n\cdot\sum_{e\in E(\hat T_r)}w(e)x_e\right),$$
thus rounding our original LP, and proving Lemma~\ref{lem:LP-rounding} and in turn Theorem~\ref{thm:pairwise-spanner-junction-tree}.

\section{Overview of Hardness Results for Additive Spanners} \label{sec:hardness}

In this section we givea high-level overview of our hardness results for additive spanners.  We begin with an informal overview of the result for $+1$-spanners, and then discuss how to extend the reduction to $+k$-spanners for larger stretch values. All details can be found in Sections~\ref{app:1-hardness} and \ref{app:k-hardness}.

In all of our reductions we will start from the \MR problem, which was first introduced by~\cite{Kor01}.  In \MR we are given a bipartite graph $G = (A,B,E)$ where $A$ is partitioned into groups $A_1, A_2, \dots, A_r$ and $B$ is partitioned into groups $B_1, B_2, \dots, B_r$, with the additional property that every set $A_i$ and every set $B_j$ has the same size (which we will call $|\Sigma|$ due to its connection to the alphabet of a $1$-round $2$-prover proof system).  This graph and partition induces a new bipartite graph $G'$ called the \emph{supergraph} in which there is a vertex $a_i$ for each group $A_i$ and similarly a vertex $b_j$ for each group $B_j$.  There is an edge between $a_i$ and $b_j$ in $G'$ if there is an edge in $G$ between some node in $A_i$ and some node in $B_j$.  A node in $G'$ is called a supernode, and similarly an edge in $G'$ is called a superedge.  

A REP-cover is a set $C \subseteq A \cup B$ with the property that for all superedges $\{a_i, b_j\}$ there are nodes $a \in A_i \cap C$ and $b \in B_j \cap C$ where $\{a,b\} \in E$.  We say that $\{a,b\}$ \emph{covers} the superedge $\{a_i, b_j\}$.  The goal is to construct a REP-cover of minimum size.  

We say that an instance of \MR is a YES instance if $OPT = 2r$ (i.e.~a single node is chosen from each group) and is a NO instance if $OPT \geq 2^{\log^{1-\epsilon} n} r$.  The following theorem, which is the starting point of our reductions, is due to Kortsarz~\cite{Kor01}:
\begin{theorem}[\cite{Kor01}] \label{thm:kor-hard}
Unless $\classNP \subseteq \classDTIME(2^{\polylog(n)})$, for any constant $\epsilon > 0$ there is no polynomial-time algorithm that can distinguish between YES and NO instances of \MR.  This is true even when the graph and the supergraph are regular, and both the supergraph degree and $|\Sigma|$ are polynomial in $2^{\log^{1-\epsilon} n}$.
\end{theorem}



\paragraph{Intuition for $+1$-Spanners:}
In the basic reduction framework, due to~\cite{Kor01,EP00-hardness}, we start with a \MR instance, and then for every group we add a vertex (corresponding to the supernode) which is connected to vertices in the group by ``connection" edges.  We then add an edge between any two supernodes that have a superedge in the supergraph.  So there is an ``outer" graph corresponding to the supergraph, as well as an ``inner" graph which is just the Min-Rep graph itself, and they are connected by connection edges.  The basic idea is that if we want stretch $3$, the only way to span a superedge is to use a path of length $3$ that goes through the \MR instance, in which case the \MR edge that is in this path corresponds to nodes in a valid REP-cover.   So if we create many copies of the outer nodes (i.e.~of the supergraph), then in a YES instance each copy can be covered relatively cheaply by using $3$-paths corresponding to the small REP-cover, while in a NO instance every copy \emph{requires} many edges simply in order to $3$-span the superedges.  This can be generalized to larger stretch values by changing the connection edges into paths of length approximately $k/2$.

Slightly more formally, suppose that we make $x$ copies of the outer nodes.  Then in the basic reduction, in a YES instance we can find a $3$-spanner that has total size of approximately $x \cdot 2r$, while in a NO instance every $3$-spanner has size at least $x \cdot 2^{\log^{1-\epsilon} n} r$.

This reduction strongly depends on having canonical paths of length at least $3$: since there is an $O(\log n)$-approximation for multiplicative stretch $2$ spanners~\cite{KP94}, no similar reduction can exist.  So if we want to prove hardness for $+1$-spanners, we need to make the true distance between supernodes at least $2$, rather than $1$.  The obvious way to do this would be to subdivide each superedge into a path of length $2$.  But now consider these new vertices: clearly in the optimal solution they must have degree at least one, and thus even in a YES instance the sparsest spanner must have size at least $x \cdot |E(G')| \approx x \cdot 2^{\log^{1-\epsilon} n} r$.  So we have entirely lost the hardness!

The intuitive solution is to subdivide superedges as before, except we use the \emph{same} middle vertex for all of the $x$ copies.  Thus we have to add only $|E(G')|$ extra nodes rather than $x \cdot |E(G')|$, and so the fact that these nodes must have degree at least $1$ in the optimal solution is no longer a problem.  Of course, now we have other problems: how do we span all of the new edges we added?  To do this we have to add yet another dummy node and paths of length $2$ from each outer node each of the middle superedge nodes, so it does not seem like we have much progress.  This seems like a catch-22: every time we add new vertices or edges to span other pairs, it becomes too expensive to span what we've added.  But now by carefully hooking these nodes up together, it turns out that we can use the same extra dummy node for each of the $x$ copies, so the total extra cost ends up being only $O(r x)$, which is still small enough that we maintain the hardness gap.

\paragraph{Hardness for $+k$-Spanners:}
To extend our hardness reduction to larger stretch values we will want to use the same basic idea: instead of supernodes being at distance $1$ from each other if they have a superedge (as in the classical reductions~\cite{Kor01,EP00-hardness,DKR16}), we will make them further away initially and make sure that canonical paths have length exactly $k$ more than the original lengths.  But here the reason we need large initial distance is very different from the reason that we needed it for $+1$-spanners: when the additive stretch is larger than $1$, then including superedges directly would lead to the possibility of spanning a superedge using non-canonical paths made up entirely of superedges.  This was the main difficulty in proving hardness of approximation for multiplicative spanners, and was overcome by~\cite{DKR16} by sparsifying the superedges so there are no short non-canonical paths (note that allowing directed edges solves this problem, for both multiplicative and additive spanners, but we want hardness for even the undirected setting).  But we pay a price in the hardness for doing this sparsification: the hardness drops to $2^{\frac{\log^{1-\epsilon} n}{k}}$, which becomes negligible when $k = \Omega(\log n)$.  

Pushing past this boundary requires giving up on sparsifying.  Thus the supergraph might have girth $4$, so if the original distance between two supernodes connected by a superedge is $d$, then the spanner instance we construct might have a path of length $3d$ using only these superedge paths.  To prevent such paths from being a problem, we will have to make them be too long, i.e.~we will need $3d - d = 2d$ to be greater than the additive stretch $k$.  But again we have the same Catch-22 as in $+1$-spanners: if we replace each superedge by a long path, then simply spanning all of the edges in those paths is too expensive.  Moreover, it now becomes difficult to span pairs that were innocuous before, e.g.~two copies of the same supernode.  Overcoming these issues requires adding even more extra paths and vertices, which create their own complications when trying to span them.  But these difficulties can be overcome with enough technical work; see Section~\ref{app:k-hardness} for details. 



\section{Hardness for $+1$-Spanners} \label{app:1-hardness}

\subsection{The Reduction}

Suppose we are given a  \MR instance $\widetilde G = (A, B, \widetilde E)$ with associated supergraph $G' = (U,V, E')$.  For any vertex $w \in U \cup V$ we let $\group(w)$ denote its group.  So $\group(u) \subseteq A$ for $u \in U$, and $\group(v) \subseteq B$ for $v \in V$.  We will assume without loss of generality that $G'$ is regular with degree $d_{G'}$ and $\widetilde G$ is regular with degree $d_{\widetilde G}$.  Let $x \in \mathbb{N}$ be a parameter which we will set later.

Our $+1$-spanner instance will have several kinds of vertices.  We first define the following vertex sets:
\begin{align*}
V_{out}^L &= U \times [x] & V_{out}^R &= V \times [x] \\
S &= \{s_y : y \in U \cup V\} & M &= \{m_{uv} : \{u,v\} \in E'\}.
\end{align*}

In other words, the $V_{out}$ vertices are just $x$ copies of the nodes in the supergraph, $S$ consists of one additional ``special" node for each node in the supergraph, and $M$ has a vertex for each superedge (these are the ``middle" nodes).  Let $V_R = V_{out}^L \cup V_{out}^R \cup A \cup B \cup S \cup M$ be the ``main" vertex set.  For technical reasons these vertices will not quite be enough, so we will also define a single special node $t$ and a set of nodes $T = \{t_y : y \in V_R\}$.  The final vertex set of our instance will be
\begin{align*}
V_G = V_R \cup \{t\} \cup T =  V_{out}^L \cup V_{out}^R \cup A \cup B \cup S \cup M \cup \{t\} \cup T.
\end{align*}

Now that our vertices are defined, we need to define edges.  We first add inner edges, which will just be a copy of the Min-Rep instance $\widetilde G$.  Formally, since $A, B \subset V_G$, we just let $E_{in} = \widetilde E$.  We will next connect the outer nodes to the inner nodes using connection edges:
\begin{align*}
E_{con} &= \{\{(u,i), a\} : u \in U \cup V \land a \in \Gamma(u) \land i \in [x]\}, 
\end{align*}

The next set of edges (the outer edges) form length-$2$ paths for each superedge:
\begin{align*}
E_{out} &= \{\{(u,i), m_{uv}\} : u \in U \land i \in [x] \land \{u,v\} \in E'\}, \\
&\bigcup \{\{(v,j), m_{uv}\} : v \in V \land j \in [x] \land \{u,v\} \in E'\}.
\end{align*}

We next connect nodes in $S$ using three edge sets: one which connects $S$ to the outer nodes, one which connects $S$ to the inner nodes, and one which connect $S$ to the middle nodes $M$:
\begin{align*}
E_{so} &= \{\{(u,i), s_u\} : u \in U \cup V \land i \in [x]\},  \\
E_{si} &= \{\{a, s_u\} : u \in U \cup V \land a \in \Gamma(u)\} \\
E_{sm} &= \{\{s_u, m_{uv}\} : u \in U \land \{u,v\} \in E'\} \bigcup \{\{s_v, m_{uv}\} : v \in V \land \{u,v\} \in E'\}.
\end{align*}

We next add  group edges to form a clique inside each group:
\begin{align*}
E_{group} &= \{\{a, b\} : \Gamma^{-1}(a) = \Gamma^{-1}(b)\}.
\end{align*}

Finally, we will add ``star" edges to make the entire graph have diameter $4$ by going through the special node $t$:
\begin{align*}
E_{star} &= \{\{t, t_y\} : y \in V_R\} \bigcup \{\{t_y, y\} : y \in V_R\}.
\end{align*}

Our final edge set $E_G$ is the union of all of these sets: $E_G = E_{in} \cup E_{con} \cup E_{out} \cup E_{so} \cup E_{si} \cup E_{sm} \cup E_{group} \cup E_{star}$.

\subsection{Analysis}

\subsubsection{Soundness}
We first analyze the soundness of this reduction: we want to show that the size any $+1$-spanner is lower bounded by $OPT_{MR}$ (the optimal solution of the Min-Rep instance $\tilde G$).  Let $H$ be an arbitrary $+1$-spanner of $G = (V_G, E_G)$.  

\begin{definition}
A path between outer nodes $(u,i)$ and $(v,j)$ is called \emph{canonical} if it has the form $(u,i) \rightarrow a \rightarrow b \rightarrow (v,j)$, where $a \in \gamma(u)$ and $b \in \Gamma(v)$.  In other words, it is a length $3$ path whose first and last edges are from $E_{con}$ and whose middle edge is from $E_{in}$.
\end{definition}

Note that if $\{u,v\}$ is a superedge then $(u,i)$ and $(v,j)$ are at distance $2$ in $G$ because of the outer edges.  Thus in $H$ they must be at distance $2$ or $3$.  The intuition is that if $\{u,v\}$ is a superedge then $H$ should contain a canonical path from $(u,i)$ to $(v,j)$ for all $i,j \in [x]$.  Such a path exists, since it corresponds to choosing nodes in a REP-cover to satisfy the $\{u,v\}$ superedge.  Unfortunately this isn't quite true, but it's ``true enough" -- with only a small loss, we can guarantee that \emph{enough} paths are canonical.  

For the purposes of lower bounding $|E(H)|$, it will turn out to be easier to only consider the case of $j=i$.  Note that $(u,i)$ and $(v,j)$ are also at distance $2$ when $j \neq i$, but for the purpose of the analysis we will only care about $j=i$.  Thus we will not try to make all possible paths canonical, but rather only when $j=i$

\begin{definition}
A $+1$-spanner $H'$ of $G$ is \emph{canonical} if for all outer nodes $(u,i)$ and $(v,i)$ where $\{u,v\} \in E'$ and $i \in [x]$, there is a canonical path between $(u,i)$ and $(v,i)$.
\end{definition}

\begin{lemma} \label{lem:canonical}
There is a canonical $+1$-spanner $H'$ of $G$ with $|E(H')| \leq 4 |E(H)|$.
\end{lemma}
\begin{proof}
For each superedge $\{u,v\} \in E'$, let $a_{uv} \in \Gamma(u)$ and $b_{uv} \in \Gamma(v)$ be arbitrary inner nodes so that $\{a_{uv},b_{uv}\} \in \widetilde E$, i.e.~so that there is a canonical path $(u,i) \rightarrow a_{uv} \rightarrow b_{uv} \rightarrow (v,j)$ for all $i,j \in [x]$.  We create a new graph $H'$ by starting with $H$, and then for each $((u,i), (v,i))$ pair of outer nodes where $\{u,v\} \in E'$ and $i \in [x]$ but where there is no canonical path from $(u,i)$ to $(v,i)$, we add the connection edges $\{(u,i), a_{uv}\}$ and $\{(v,i), b_{uv}\}$ and the inner edge $\{a_{uv}, b_{uv}\}$ (if these edges do not already exist).  Since $H'$ contains $H$, we know that $H'$ is also a $+1$-spanner, and it is canonical by construction.

So it remains to prove the size bound.  How many edges did we add to $H$ to get $H'$?  Suppose we added (at most $3$) edges to canonically span the pair $((u,i), (v,i))$.  Then in $H$ the path between $(u,i)$ and $(v,i)$ already had length at most $3$, but was not a canonical path.  Since it was not a canonical path but had length $3$, it must have one of the following forms:
\begin{enumerate}
\item $(u,i) \rightarrow m_{uv} \rightarrow (v,i)$, or
\item $(u,i) \rightarrow m_{uv} \rightarrow s_v \rightarrow (v,i)$, or
\item $(u,i) \rightarrow s_u \rightarrow m_{uv} \rightarrow (v,i)$.
\end{enumerate}

This classification is easy to see by inspection.  The unique path of type $1$ is clearly the only path of length $2$.  And the only nodes reachable in $2$ hops from $(u,i)$ are inner nodes, outer nodes of the form $(u,j)$, $s_u$, $m_{uv'}$ (where $v'$ is not necessarily equal to $v$), $s_{v'}$, $t_{(u,i)}$, and $t$.  If a $2$-hop path to an inner node can be extended in one hop to $(v,i)$ then we have a canonical path, so by assumption no such path exists.  Of the other nodes reachable in two hops, the only ones adjacent to $(v,i)$ are $m_{uv}$ and $s_v$.  Hence no paths other than paths 2 and 3 exist of length $3$.

This classification implies that either $\{(u,i), m_{uv}\}$ or $\{(v,i), m_{ab}\}$ (or both) is an edge in $E(H)$.  Let $e_{uvi}$ be one of these edges (if both exist, choose one arbitrarily).  When creating $H'$, we added up to three edges in order to add a canonical path between $(u,i)$ and $(v,i)$.  We will charge those edges to $e_{uvi}$.  

When we do this for all $\{u,v\} \in E'$ and $i \in [x]$, it is clear that we have charged a different edge for each such (superedge, value) pair.  Since each charge involves at most three new edges, and every existing edge is charged at most once, we immediately get the lemma.
\end{proof}

\begin{theorem} \label{thm:soundness}
Any $+1$-spanner $H$ of $G$ has $|E(H)| \geq x \cdot OPT_{MR} / 4$.
\end{theorem}
\begin{proof}
By Lemma~\ref{lem:canonical}, we just need to show that $|E(H')| \geq x \cdot OPT_{MR}$.  

For $i \in [x]$, let $C_i \subseteq A \cup B$ be the set of inner nodes which in $H'$ are adjacent to the associated outer node, i.e.~$C_i = \{a \in A \cup B : \{(\Gamma^{-1}(a), i), a\} \in E(H')\}$.  We first claim that $C_i$ is a REP-cover of $\tilde G$ for all $i \in [x]$.  To see this, fix $i \in [x]$ and a superedge $\{u,v\} \in E'$.  Since $H'$ is canonical there is a canonical path $(u,i) \rightarrow a \rightarrow b \rightarrow (v,i)$ in $H'$.  Thus $a$ and $b$ are both in $C_i$, and since the middle edge of the canonical path is an edge in $\widetilde E$ we know that $C_i$ covers the superedge $\{u,v\}$.  Hence $C_i$ is a valid REP-cover for all $i \in [x]$, and thus $|C_i| \geq OPT_{MR}$

Now note that by definition, every node $a \in C_i$ is incident on an edge $\{a, (\Gamma^{-1}(a), i)\}$.  Since this edge is different for all $i \in [x]$ and $a \in C_i$, we immediately get that $|E(H')| \geq \sum_{i=1}^x |C_i| \geq x \cdot OPT_{MR}$.  
\end{proof}

\subsubsection{Completeness}

We now show completeness: that there is a $+1$-spanner of $G$ which does not cost too much \emph{more} than $OPT_{MR}$ in the YES case.  Recall that in a YES instance there is a valid REP-cover of size $|U| + |V|$ where for each supernode we choose exactly one representative.  Let $C$ be such a REP-cover, for each $u \in U$ let $a_u$ be the unique element of $C \cap \Gamma(u)$, and for each $v \in V$ let $b_v$ be the unique element of $C \cap \Gamma(v)$.  

The subgraph $H$ of $G$ we will analyze will include $E_{in} \cup E_{so} \cup E_{si} \cup E_{sm} \cup E_{star}$.  It will also include connection edges $\{(u,i), a_u\}$ for all $i \in [x]$ and $u \in U$ and connection edges $\{(v, i), b_v\}$ for all $i \in [x]$ and $v \in V$.  Finally, it will include a star of group edges in each group, with the representative from $C$ as the center.  More formally, for each $u \in U$ we will include the group edges $\{\{a_u, a\} : a \in \Gamma(u)\}$ and for each $v \in V$ we will include the group edges $\{\{b_v, b\} : b \in \Gamma(v)\}$.  

We first analyze the size of $H$, and then later will prove that it is a $+1$-spanner of $G$.  Recall that $d_{G'}$ is the degree of the supergraph $G'$ (which is regular).  Let $n'$ be the number of nodes in the supergraph, and let $|\Sigma|$ denote the size of each group in $\widetilde G$.  

\begin{theorem} \label{thm:completeness-size}
$H$ has at most $8 n' d_{G'}|\Sigma|^2 + 4n' x$ edges when $\widetilde G$ is a YES instance.
\end{theorem}
\begin{proof}
Clearly $E_{in} = \widetilde E$ has at most $n' d_{G'} |\Sigma|^2$ edges,  $E_{so}$ has at most $n' x$ edges, $E_{si}$ has at most $n' |\Sigma|$ edges, $E_{sm}$ has at most $n' d_{G'}$ edges, and $E_{star}$ has at most $2 |V_R| = 2 (n' x + n' + n' d_{G'})$ edges. 
The number of group edges added is clearly at most $n' |\Sigma|$, and the number of connection edges is at most $n' x$.  Putting this together, we get that
\begin{align*}
|E(H)| &\leq n' d_{G'} |\Sigma|^2 + n' x + n' |\Sigma| + n' d_{G'} + 2(n' x + n' + n' d_{G'}) + n' |\Sigma| + n' x \\
& \leq 8 n' d_{G'} |\Sigma|^2 + 4 n' x
\end{align*}
as claimed.
\end{proof}

\begin{theorem} \label{thm:completeness-correct}
$H$ is a $+1$-spanner of $G$.
\end{theorem}
\begin{proof}
Since $H$ includes all star edges it has diameter at most $4$, and hence any two vertices which were originally at distance $3$ or $4$ are spanned by $H$.  Thus we only need to analyze pairs which were at distance $1$ or $2$ in $G$.  

We begin by analyzing pairs which were at distance $1$ in $G$, i.e.~edges of $G$.  If some edge is also in $H$ then it is trivially spanned.  The only edges of $G$ not in $H$ are the outer edges, some group edges, and some connection edges.  Each outer edge has a length $2$ path spanning it by using one edge of $E_{so}$ and one edge of $E_{sm}$.  Each group edge has a length $2$ path spanning it using the star inside every group that is included in $H$.  Similarly, each connection edge has a $2$-path spanning it using the included connection edge and an included star edge.  Slightly more formally, consider a connection edge $\{(u,i), a\}$.  Then $H$ includes the connection edge $\{(u, i), a_u\}$ and the group edge $\{a_u, a\}$.  The same analysis holds for connection edges on the other side.

We now consider pairs which were at distance $2$ in $G$, and prove that they are at distance at most $3$ in $H$.  To do this, first note that if the $2$-path connecting two nodes in $G$ has both edges in $H$ then it is trivially still spanned in $H$, and if one of the two edges are in $H$ then it still has distance at most $3$ in $H$ (since by the previous analysis all missing edges are replaced by a path of length $2$).  So we only need to worry about pairs of vertices at distance $2$ in $G$ where both edges on the $2$-path in $G$ are not in $H$.  

This leaves only a few cases to analyze, based on the types of the two edges in the path.  We consider them each.

\begin{enumerate}
\item Two outer edges.  From an outer node $(u,i) \in V_{out}^L$, any path involving two outer edges must end at either an outer node $(v,j)$ where $\{u,v\} \in E'$ or at a node $(u,j)$ with $j \neq i$.  For the first of these, we know by construction that there is a length-$3$ canonical path from $(u,i)$ to $(v,j)$.  For the second, there is still a $2$-path by going through $s_u$ and using two edges from $E_{so}$.  The same analysis holds for paths from outer nodes $(v,j) \in V_{out}^R$.  

The only other type of $2$-path involving only outer edges are path between two middle nodes, i.e.~$2$-paths from $m_{uv}$ to $m_{u' v'}$.  For such a path to exist, either $u' = u$ or $v' = v$, and hence $H$ still has a $2$-path of the form $m_{uv} \rightarrow s_u \rightarrow m_{uv'}$ or of the form $m_{uv} \rightarrow s_v \rightarrow m_{u'v}$.

\item Two connection edges.  Any $2$-path involving two missing connection edges must be either between outer nodes $(u,i)$ and $(u,j)$ with $i \neq j$, or between two inner nodes $a,b$ with $\Gamma^{-1}(a) = \Gamma^{-1}(b)$.  In the first case there is still a $2$-path between the nodes by using edges in $E_{so}$, and in the second case there is still a $2$-path by using the remaining edges of $E_{group}$ in $\Gamma(\Gamma^{-1}(a))$.

\item Two group edges.  This is trivial since any pair of nodes connected by such a $2$-path are also connected by a single group edge, so the analysis for a single edge holds.

\item Outer edge and connection edge.  Any such path is from a middle node $m_{uv}$ to an inner node $a \in \Gamma(u) \cup \Gamma(b)$.  A $2$-path between these nodes exists in $H$ by using an edge in $E_{sm}$ and an edge in $E_{si}$.  

\item Outer edge and group edge.  No such path exists.

\item Connection edge and group edge.  The two endpoints of such a path must be an outer node $(u,i)$ and an  inner node $a \in \Gamma(u)$.  But $G$ already has a connection edge between them, so they are at distance $1$ in $G$ and the previous analysis applies.  
\end{enumerate}

Thus $H$ is a $+1$-spanner of $G$.  
\end{proof}

\subsubsection{Putting it Together}

We can now finish the proof of hardness.

\begin{theorem}
For any constant $\epsilon > 0$, there is no polynomial-time $2^{\log^{1-\epsilon} n}$-approximation algorithm for the $+1$-spanner problem unless $NP \subseteq DTIME(2^{\polylog(n)})$.  
\end{theorem}
\begin{proof}
We first instantiate $x$ to be $d_{G'}|\Sigma|^2$.  With this setting of $x$, it is obvious that the graph $G$ created by our reduction has size $n$ which is polynomial in the size $n' |\Sigma|$ of the Min-Rep instance $\widetilde G$ (and that we can create $G$ in polynomial time).  

Theorem~\ref{thm:kor-hard} implies that under the given complexity assumption, there is no polynomial-time algorithm which can distinguish between the YES case where $OPT_{MR} = n'$ and the NO case where $OPT_{MR} \geq n' 2^{\log^{1-\epsilon}(n'|\Sigma|}$.  Let $OPT$ denote the size of the optimal $+1$-spanner for $G$.  Then Theorem~\ref{thm:soundness} implies that if $\widetilde G$ is a NO instance then $OPT \geq x \cdot OPT_{MR} / 4 \geq x n' \cdot 2^{\log^{1-\epsilon}(n'|\Sigma|)}$ / 4.  On the other hand, if $\widetilde G$ is a YES instance, then Theorems~\ref{thm:completeness-size} and~\ref{thm:completeness-correct} imply that $OPT \leq 12 n' x$.  Thus we cannot approximate the $+1$-spanner problem better than $48 \cdot 2^{\log^{1-\epsilon}(n |\Sigma|)}$.  By using a sufficiently smaller value of $\epsilon$ (and restricting our attention to large enough input instances), this implies hardness of $2^{\log^{1-\epsilon} n}$ as claimed.  
\end{proof}

\section{Hardness for $+k$-Spanners} \label{app:k-hardness}

We now prove hardness for larger additive stretch, with the goal of proving Theorem~\ref{thm:hardness-main}. 

\subsection{The Reduction}

Suppose we are given a bipartite Min-Rep instance $\widetilde G = (A, B, \widetilde E)$ with associated supergraph $G' = (U,V, E')$.  For any vertex $w \in U \cup V$ we let $\group(w)$ denote its group.  So $\group(u) \subseteq A$ for $u \in U$, and $\group(v) \subseteq B$ for $v \in V$.  We will assume without loss of generality that $G'$ is regular with degree $d$, and let $n' = |U \cup V|$.  Let $x \in \mathbb{N}$ be a parameter which we will set later.

Our $+k$-spanner instance will have several kinds of vertices.  We first define the following vertex sets:
\begin{align*}
V_{out}^L &= U \times [x] \times [k-1] & V_{out}^R &= V \times [x] \times [k-1] \\
S &= \{s_y : y \in U \cup V\} & L &= \{\ell_{u, i, j} : u \in U \cup V, i \in [x], j \in [k-1]\} \\
M &= \{m_{uv}^i : \{u,v\} \in E' \land i \in [k-2]\} & V_{in} &= A \cup B \\
Q &= \{q_y : y \in U \cup V\} & P &= \{p_{u,i,j} : u \in U \cup V, i \in [x], j \in \left[\left\lceil \frac{k-1}{2}\right\rceil\right]\}
\end{align*}

Let the union of the above vertices be $V_R = V_{out}^L \cup V_{out}^R \cup A \cup B \cup S \cup M \cup L \cup P \cup Q$, which will be the ``main" vertex set.  Note that $|V_R| \leq k n' x + n' |\Sigma| + n' + n' k x + n' d k +n' + n' k x \leq 3 n' k x + n'(|\Sigma| + 2 + dk) \leq 3n' k x + 2n'(|\Sigma| + dk)$ 

In order to decrease the diameter, we will also have an extra node $t$ and extra nodes
\begin{align*}
T &= \{t_{y, i} : y \in V_R \land i \in [k-1]\}.
\end{align*}

The final vertex set of our instance will be
\begin{align*}
V_G = V_R \cup T =  V_{out}^L \cup V_{out}^R \cup A \cup B \cup S \cup M \cup L \cup Q \cup P \cup T \cup \{t\}.
\end{align*}

Now that our vertices are defined, we need to define edges.  We first add inner edges, which will just be a copy of the Min-Rep instance $\widetilde G$.  Formally, since $A, B \subset V_G$, we just let $E_{in} = \widetilde E$.  

We next define the crucial connection edges, we connect the inner nodes to the outer nodes:
\begin{align*}
E_{con} = \{\{(u,i,1), a\} : u \in U \cup V \land a \in \group(u) \land i \in [x]\}.
\end{align*}

We now need to define paths among the outer nodes:
\begin{align*}
E_{path} = \{\{(u,i,j), (u,i,j+1)\} : u \in U \cup V \land i \in [x] \land j \in [k-2]\}
\end{align*}

The next edges also form paths: $E_P$ is set up to provide alternate bounded-length paths between the outermost outer nodes, and $E_L$ is set up to provide bounded-length paths from the outermost outer nodes to $S$. 
\begin{align*}
E_L &= \{\{(u,i, k-1), \ell_{u, i, 1}\} : u \in U \cup V \land i \in [x]\} \\
&\cup \{\ell_{u,i,j}, \ell_{u,i,j+1}\} : u \in U \cup V \land i \in [x] \land j \in [k-2]\} \\
&\cup \{\{\ell_{u,i, k-1}, s_u\} : u \in U \cup V \land i \in [x] \\
E_P &= \{\{(u,i,k-1), p_{u,i,1}\} : u \in U \cup V \land i \in [x]\} \\
& \cup \{\{p_{u,i,j}, p_{u,i,j+1}\} : u \in U \cup V \land i \in [x] \land j \in \left\lceil \frac{k-1}{2} \right\rceil - 1\} \\
&\cup \{\{p_{u,i, \left\lceil \frac{k-1}{2} \right\rceil}, q_u\} : u \in U \cup V \land i \in [x]\}
\end{align*}

We now connect the nodes in $S$ using two edge sets:
\begin{align*}
E_{so} &= \{\{s_u, \ell_{u,i,k-1}\} : u \in U \cup V \land i \in [x]\}\\
E_{sm} &= \{\{s_u, m_{uv}^1\} : u \in U \land \{u,v\} \in E'\} \cup \{\{s_v, m_{uv}^{k-2}\} : v \in V \land \{u,v\} \in E'\}.
\end{align*}

And now we connect the nodes in $M$ via more paths.  Note that the following edge set is empty if $k=3$.
\begin{align*}
E_M &= \{\{m_{uv}^i, m_{uv}^{i+1}\} : \{u,v\} \in E' \land i \in [k-3]\}.
\end{align*}

We connect $M$ to $V_{out}$ using outer edges:
\begin{align*}
E_{out} &= \{\{(u,i,k-1), m_{uv}^1\} : u \in U \land i \in [x] \land \{u,v\} \in E'\} \\
& \cup \{\{(v,i,k-1), m_{uv}^{k-2}\} : v \in V \land i \in [x] \land \{u,v\} \in E'\}.
\end{align*}

We now add group edges to form a clique in each group, and star edges to make the entire graph have diameter $2k$:
\begin{align*}
E_{group} &= \{\{a,b\} : a,b \in A \cup B \land \group^{-1}(a) = \group^{-1}(b)\}, \\
E_{star} &= \{\{t, t_{y, k-1}\} : y \in V_R\} \cup \{\{t_{y, i}, t_{y, i+1}\} : y \in V_R \land i \in [k-2]\}.
\end{align*}

This completes the reduction: the final edge set is 
\begin{align*}
E_G &= E_{in} \cup E_{con} \cup E_{path} \cup E_L \cup E_P \cup E_{so} \cup E_{sm} \cup E_M \cup E_{out} \cup E_{group} \cup E_{star}.
\end{align*}

\subsection{Analysis}

\subsubsection{Soundness}

We will first show that the size of any $+k$-spanner of $G$ is lower bounded by $OPT_{MR}$ (the optimal solution of the Min-Rep instance $\widetilde G$).  Let $H$ be an arbitrary $+k$-spanner of $G$.

\begin{definition}
Let $\{u,v\} \in E'$.  A path between outer nodes $(u, i, k-1)$ and $(v,j,k-1)$ is called \emph{canonical} if it has the form
\begin{align*}
(u,i,k-1) \rightarrow (u,i,k-2) \rightarrow \dots \rightarrow (u,i,1) \rightarrow a \rightarrow b \rightarrow (v,j,1) \rightarrow (v,j,2) \rightarrow \dots \rightarrow (v,j,k-1).
\end{align*}
In other words, it goes through edges in $E_{path}$, then through a connection edge, then an inner edge, then a connection edges, and then more edges in $E_{path}$.
\end{definition}

Note that any canonical path has length exactly $2k-1$, and that $a \in \group(u)$ and $b \in \group(v)$.  Since the distance in $G$ from $(u,i,k-1)$ and $(v,j,k-1)$ is $k-1$ via the path $(u,i,k-1) \rightarrow m_{uv}^1 \rightarrow m_{uv}^2 \rightarrow \dots \rightarrow m_{uv}^{k-2} \rightarrow (v,j,k-1)$, a canonical path is a valid spanning path for this pair.  

\begin{definition}
A $+k$-spanner $H'$ of $G$ is \emph{canonical} if for all outer nodes $(u,i,k-1)$ and $(v,i,k-1)$ with $\{u,v\} \in E'$, there is a canonical path between $(u,i,k-1)$ and $(v,i,k-1)$
\end{definition}

Note that this definition of canonical uses the same value of $i$ on both sides -- this is to simplify the charging argument used in the proof of the next lemma.  

\begin{lemma} \label{lem:k-canonical}
There is a canonical $+k$-spanner $H'$ of $G$ with $|E(H')| \leq 2k |E(H)|$.
\end{lemma}
\begin{proof}
For each superedge $\{u,v\} \in E'$, let $a_{uv} \in \group(u)$ and $b_{uv} \in \group(v)$ be arbitrary inner nodes so that $\{a_{uv}, b_{uv}\} \in \widetilde E$, i.e.~an pair of nodes which would cover the superedge.  We create a new graph $H'$ by starting with $H$, and then for each pair of outer nodes $(u,i,k-1)$ and $(v,i,k-1)$ where $\{u,v\} \in E'$ that do not have a canonical path, we add the canonical path between them through $a_{uv}$ and $b_{uv}$.  Slightly more formally, if any edges in the following path are missing, we add them (note that all of these edges are in $G$): $(u,i,k-1) \rightarrow (u,i,k-1) \rightarrow \dots \rightarrow (u,i,1) \rightarrow a_{uv} \rightarrow b_{uv} \rightarrow (v,i,1) \rightarrow (v,i,2) \rightarrow \dots \rightarrow (v,i,k-1)$.  

Since $H'$ includes $H$, it is clearly a $+k$-spanner.  It is canonical by construction: if it is not canonical then there is a pair of outer nodes $(u,i,k-1)$ and $(v,i,k-1)$ with $\{u,v\} \in E'$ which is not spanned by a canonical path.  But for every such pair we added a canonical path.  

Thus it remains only to bound the size of $H'$.  We do this by a charging argument: we will show how to charge edges we added to already existing edges of $H$, in such a way that no edge of $H$ is charged more than $2k-1$ times.  This will clearly prove the lemma.  

Suppose to get $H'$ we added a canonical path to span $(u,i, k-1)$ and $(v,i,k-1)$.  Note that this path has length exactly $2k-1$.  In order to charge this path to an edge of $H$, we first need to understand the way in which $(u,i,k-1)$ and $(v,i,k-1)$ could have been spanned in $H$.  Since the distance between them in $G$ is exactly $k-1$ (using the path through $M$), in $H$ there must be a path between them of length at most $2k-1$, and this path must be non-canonical (or else we would not have added a canonical path).  It is easy to verify that there are only three non-canonical paths between them of length at most $2k-1$:
\begin{enumerate}
\item $(u,i,k-1) \rightarrow m_{uv}^1 \rightarrow m_{uv}^2 \rightarrow \dots \rightarrow m_{uv}^{k-2} \rightarrow (v,i,k-1)$ (length $k-1$), and
\item $(u,i,k-1) \rightarrow \ell_{u,i,1} \rightarrow \ell_{u,i,2} \rightarrow \dots \rightarrow \ell_{u,i,k-1} \rightarrow s_u \rightarrow m_{uv}^1 \rightarrow m_{uv}^2 \rightarrow \dots \rightarrow m_{uv}^{k-2} \rightarrow (v,i,k-1)$ (length $2k-1$), and 
\item $(u,i,k-1) \rightarrow m_{uv}^1 \rightarrow m_{uv}^2 \rightarrow \dots \rightarrow m_{uv}^{k-2} \rightarrow s_v \rightarrow \ell_{v,i,k-1} \rightarrow \ell_{v,i,k-2} \rightarrow \dots \rightarrow \ell_{v,i,1} \rightarrow (v,i,k-1)$ (length $2k-1$).  
\end{enumerate}

Hence $H$ must include at least one of these paths.  All three of these paths include an outer edge: either $\{(u,i,k-1), m_{uv}^1\}$ or $\{(v,i,k-1), m_{uv}^{k-2}\}$ or both.  We charge all edges on the canonical path we added to whichever of these outer edges exists in $H$ (if they both exist in $H$, then we pick one arbitrarily).  Clearly any two canonical paths will be charged to different outer edges, and thus each edge in $H$ is charged at most $2k-1$ times, proving the lemma.
\end{proof}

This now makes it easy to prove soundness.

\begin{theorem} \label{thm:k-soundness}
Any $+k$-spanner $H$ of $G$ has $|E(H)| \geq \frac{x}{2k} \cdot OPT_{MR}$
\end{theorem}
\begin{proof}
By Lemma~\ref{lem:k-canonical}, we just need to show that $|E(H')| \geq x \cdot OPT_{MR}$.

For $i \in [x]$, let $C_i \subseteq A \cup B$ be the set of inner nodes which in $H'$ are adjacent to the associated outer node, i.e.~$C_i = \{a \in A \cup B : \{(\Gamma^{-1}(a), i,1), a\} \in E(H')\}$.  We first claim that $C_i$ is a REP-cover of $\tilde G$ for all $i \in [x]$.  To see this, fix $i \in [x]$ and a superedge $\{u,v\} \in E'$.  Since $H'$ is canonical there is a canonical path $(u,i,k-1) \rightarrow (u,i,k-2) \rightarrow \dots \rightarrow (u,i,1) \rightarrow a \rightarrow b \rightarrow (v,i,1) \rightarrow (v,i,2) \rightarrow (\dots \rightarrow (v,i,k-1)$ in $H'$.  Thus $a$ and $b$ are both in $C_i$, and since the middle edge of the canonical path is an edge in $\widetilde E$ we know that $C_i$ covers the superedge $\{u,v\}$.  Hence $C_i$ is a valid REP-cover for all $i \in [x]$, and thus $|C_i| \geq OPT_{MR}$

Now note that by definition, every node $a \in C_i$ is incident on a connection edge $\{a, (\Gamma^{-1}(a), i,1)\}$.  Since this edge is different for all $i \in [x]$ and $a \in C_i$, we immediately get that $|E(H')| \geq \sum_{i=1}^x |C_i| \geq x \cdot OPT_{MR}$.  
\end{proof}

\subsubsection{Completeness}

We now show completeness: that there is a $+1$-spanner of $G$ which does not cost too much \emph{more} than $OPT_{MR}$.  Let $C$ be a REP-cover of size $OPT_{MR}$, and for each $u \in U \cup V$ let $C_u = C \cap \group(u)$.  

The subgraph $H$ of $G$ that we will analyze will include $E_{in} \cup E_{path} \cup E_{L} \cup E_P  \cup E_{so} \cup E_{sm} \cup E_M \cup E_{group} \cup E_{star}$.  In other words, the only edge sets we defined which it does not entirely include are $E_{con}$ and $E_{out}$.  Of these, we do not include any edges of $E_{out}$. We include connection edges to inner nodes in $C$, i.e.~we also include all edges in the set $\{\{(u,i,1), a\} : u \in U \cup V \land a \in C_u \land i \in [x]\}$.  

We first analyze the size of $H$, and the later will prove that it is a $+k$-spanner of $G$.  Recall that $d$ is the degree of the supergraph $G'$ (which is regular), and $|\Sigma|$ is the size of each group.  Let $n' = |U \cup V|$.

\begin{theorem} \label{thm:k-completeness-size}
$H$ has at most $ x \cdot 7k^2 OPT_{MR} + 8 k^2 n' d |\Sigma|^2$  edges
\end{theorem}
\begin{proof}
The following size bounds are direct from the definitions:
\begin{align*}
|E_{in}| &\leq n' d |\Sigma|^2, & |E_{path}| & \leq (k-2) n' x, & |E_L| &\leq k n' x, \\
|E_{so}| &\leq n'x, & |E_{sm}| &\leq n' d, &|E_M| &\leq k n' d, \\
|E_{group}| &\leq n' |\Sigma|^2, & |E_{star}| &\leq 3n' k^2 x + 2n'k(|\Sigma| + dk) &|E_P|&\leq k n' x
\end{align*}

The number of connection edges we add is at most $x \cdot |C| = x \cdot OPT_{MR}$.  Adding these all up, we get that 
\begin{align*}
|E(H)| &\leq x(OPT_{MR} + 7 k^2 n') + 6 k n' |\Sigma|^2 + 2n'k^2 d \\
& \leq x(OPT_{MR} + 7 k^2 n') + 8 k^2 n' d |\Sigma|^2 \\
& \leq x \cdot 8 k^2 OPT_{MR} + 8 k^2 n' d |\Sigma|^2
\end{align*} 
where for the last inequality we used the fact that $OPT_{MR} \geq n'$.  
\end{proof}

Note that this implies that if $x  \geq d |\Sigma|^2$ then $|E(H)| \leq x \cdot 8 k^2 OPT_{MR} + x \cdot 8 k^2  OPT_{MR} \leq x \cdot 16 k^2 OPT_{MR}$

\begin{theorem} \label{thm:k-completeness-correct}
$H$ is a $+k$-spanner of $G$.
\end{theorem}
\begin{proof}
Because of $E_{star}$, the diameter of $H$ is at most $2k$ and thus we only need to worry about spanning a pair of nodes if their distance in $G$ is at most $k-1$.  

We first consider the simple case of nodes at distance $1$ (i.e.~edges of $G$).  The only edges of $G$ missing from $H$ are the outer edges and some connection edges.  Consider a connection edge $\{(u,i,1), a\}$ which is not in $H$.  Then since $C_u \neq \emptyset$ there is some $a' \in C_u$, and thus $H$ has a path of length $2$ spanning the missing edge: $(u,i,1) \rightarrow a' \rightarrow a$, where the first is a connection edge and the second is a group edge.  

Now consider a missing outer edge $\{(u,i,k-1), m_{uv}^1\}$.  In $H$ there is a path of length $k+1$ between these nodes by using edges of $E_L$ and $E_{sm}$, in particular the path $(u,i,k-1) \rightarrow \ell_{u,i,1} \rightarrow \ell_{u,i,2} \rightarrow \dots \rightarrow \ell_{u,i,k-1} \rightarrow s_u \rightarrow m_{uv}^1$.  A similar path exists for missing outer edges of the form $\{(v,i,k-1), m_{uv}^{k-2}\}$.  

For the next case, consider two nodes at distance $2$ in $G$.  If the shortest path between them has zero or one edge in $E(G) \setminus E(H)$, then we know that their distance in $H$ is at most $k+2$ as desired.  The only pairs of nodes at distance $2$ in $G$ whose shortest path contains two edges from $E(G) \setminus E(H)$ are 1) outermost outer nodes corresponding to the same supernode, e.g.~nodes $(u,i, k-1)$ and $(u,i', k-1)$ with $i \neq i'$, and 2) innermost outer nodes corresponding to the same supernode, e.g.~nodes $(u,i,1)$ and $(u,i', 1)$ with $i \neq i'$.  The second case is simple: these nodes have a $2$-path in $H$ using a different shortest path of two connection edges to the same inner node in $C_u$.  The first case is essentially why we added the $P$ and $Q$ nodes: there is a path of length either $k+1$ (if $k$ is even) or $k+2$ (if $k$ is odd) between $(u,i, k-1)$ and $(u,i', k-1)$ by using edges in $E_P$. 

Now consider two nodes who are at distance more than $2$ (but less than $k$) in $G$.  If their shortest path contains at most $1$ outer edge or connection edge, then the above discussion implies that they are spanned in $H$.  Since they are at distance less than $k$, their shortest path cannot include \emph{both} an outer edge and a connection edge.  Hence we are only concerned with pairs whose shortest path contains more than one connection edge or more than one outer edge.  

It is easy to verify that every shortest path of length less than $k$ includes at most two connection edges, or at most two outer edges.  Pairs that are connected by a shortest path with two connection edges are spanned in $H$ by a path of length at most $2$ longer than in $G$ by routing around the missing connection edges.  

The more interesting case is pairs of nodes that are connected by a shortest path which includes two outer edges.   If these two outer edges directly follow each other on the shortest path, then we already showed that they are spanned by a path of length at most $k+2$.  Hence using this detour only add at most $+k$ to the length of the shortest path and thus $H$ has a path at most $k$ longer than in $G$.  On the other hand, if they do not directly follow each other on the shortest path, then the fact that the shortest path has length less than $k$ implies that in fact the two nodes are outermost outer nodes $(u,i,k-1)$ and $(v,j,k-1)$ where $\{u,v\} \in E'$ (no other pairs have distance between $3$ and $k-1$ in $G$ with shortest path containing two non consecutive outer edges).  By construction, there is a canonical path of length $2k-1$ in $H$ between these nodes, and in $G$ their distance is $k-1$, so $H$ does span them.  

Since this exhausts all cases, we have shown that $H$ is a $+k$-spanner of $G$.  
\end{proof}

\subsubsection{Putting it Together}

\begin{theorem}
For any constant $\epsilon > 0$ and any value $k \geq 3$ (not necessarily constant), there is no polynomial-time $2^{\log^{1-\epsilon} n} / k^3$-approximation algorithm for the $+k$-spanner problem unless $NP \subseteq DTIME(2^{\polylog(n)})$.
\end{theorem}
\begin{proof}
We first instantiate $x$ to $d |\Sigma|^2$.  This means that the size of our reduction is 
\begin{align*}
|V(H)| &\leq k |V_R| \leq k(3n'kx + 2n'(|\Sigma| + dk)) \leq 3n'k^2 d|\Sigma|^2 + 2n'|\Sigma| + 2n'dk \\
& \leq 7 k^2 n' d |\Sigma|^2 \leq 7k^2 (n' |\Sigma|)^2.
\end{align*}

Let $OPT$ denote the size of the sparsest $+k$-spanner of $H$.  Then Theorems~\ref{thm:k-soundness}, \ref{thm:k-completeness-size}, and \ref{thm:k-completeness-correct} imply that 
\begin{align*}
\frac{x}{2k} \cdot OPT_{MR} \leq OPT \leq x \cdot 16k^2 OPT_{MR}.
\end{align*}

Putting the size and the approximation together (with the fact that the original Min-Rep instance had size $n' |\Sigma|$), we get that an $f(n)$-approximation for the $+k$-spanner problem would result in a $32k^3 \cdot f(7k^2 n^2)$-approximation for Min-Rep.  Then Theorem~\ref{thm:kor-hard} implies that this is at least $2^{\log^{1-\epsilon} n}$.  Using a smaller value of $\epsilon$ in order to dominate the constant factors completes the proof.
\end{proof}

\bibliographystyle{plain}
\bibliography{refs}

\appendix


\newcommand{\Metric}[1]{\mathcal{M}_{#1}}

\section{The Proof of Height Reduction (Lemma~\ref{lem:height-reduction})}
\label{sec:dist-height-reduction}

In this section, we sketch the proof of the existence of the tree
$\hat{T}_{r}$ in Lemma~\ref{lem:height-reduction}.
whose subtree corresponds to an
arborescence in the original graph with a bounded cost.
In short, the tree $\hat{T}_{r}$ in Lemma~\ref{lem:height-reduction}
in constructed by listing all paths rooted at $r$ of length at most $\sigma$  
in the metric completion of $G$ and then form a suffix tree on these paths. 
Our reduction is, indeed, the same as ``path-splitting'' technique
in \cite{CEGS11}.
All we need is to show the cost guarantee using 
the famous Zelikovsky's Height Reduction Theorem.

Given a graph $G$, the {\em metric-completion}
graph of $G$, denoted by $\Metric{G}$,
is a complete graph on the same vertex set as $G$,
where each edge $(u,v)$ in $\Metric{G}$ has weight 
equal to the distance of $u,v$ in $G$.
Thus, each edge $(u,v)$ of $\Metric{G}$ is associated
with a shortest $u,v$-path in $G$, denoted by $\phi(u,v)$.

\begin{theorem}[Zelikovsky's Height Reduction Theorem \cite{Z97,HRZ01}]
\label{thm:height-reduction}
For any arborescence $J$ with edge weights $w:E(J)\rightarrow\reals^{\geq 0}$, 
there exists an arborescence $J'$ 
in the metric completion $\Metric{J}$ of $J$ 
with height $\sigma$ 
and edge weights $w':E(J')\rightarrow\reals^{\geq 0}$ such that 
$\mathcal{L}(J)=\mathcal{L}(J')$ where $\mathcal{L}(J)$ and $\mathcal{L}(J')$
are the sets of leaves of $J$ and $J'$, respectively, and 
$w(J) = O(\sigma |\mathcal{L(J)}|^{1/\sigma})$.
\end{theorem}
\begin{proof}[Sketch of The Proof]
We order vertices of $J$ by depth-first-search order. 
Let $i=0$, and let $V_0=\mathcal{L}(J)$ be the set of ``leaves'' of $J$.
We first add a copy of $V_i$ to $\hat{J}$.
Then we partition vertices of $V_i$ into blocks $B_1,\ldots,B_q$,
each of size $\Delta$, according to the depth-first-search order.
For each block $B_j$, 
we find the least common ancestor of vertices of $B_j$, denoted by $a_j$.
We add edges from each $v\in B_j$ to $a_j$
(which will be added to $J'$).
Then we define $V_{i+1}=\{a_j\}_j$ (remove duplicate)
and continue the same process on $V_{i+1}$ until $|V_{i+1}|=1$.
At the termination, if the only vertex left in $V_{i+1}$ is $v\neq r$, 
then we add a copy of $r$ to $J'$
and join $v$ to $r$.
We define the weight $w'(u,v)$ of an edge $(u,v)\in E(J')$ to be 
the weight of the unique $u,v$-path in $J$.

In each step, we pay a factor of $O(\Delta\cdot w(J))$ because
paths from the same block may share edges,
and they can also share edges with one block from 
the left and one from the right.
Since there are only $\log_{\Delta}(|\mathcal{L}(J)|)$ steps,
we pay a factor of 
$O(\Delta\cdot\log_{\Delta}(|\mathcal{L}(J)|))$. 
By adjusting the parameter $\Delta$ so that
$\sigma=\log_{\Delta}(|\mathcal{L}(J)|)$
and thus $\Delta=|\mathcal{L}(J)|^{1/\sigma}$, 
we have the claimed upper bound.
\end{proof}

\paragraph{Building the tree $\hat{T}_{r}$.}
Next, we construct the tree $\hat{T}_{r}$ as 
in Lemma~\ref{lem:height-reduction} by listing all paths of length 
(number of edges) at most $\sigma$ in the metric completion of $G$
(denoted by $\Metric{G}$) that start from $r$ and
then form a suffix tree.
Let $\mathbb{P}_{\sigma}^r$ denote the set of all paths in
$\Metric{G}$ that start at $r$ and have length at most $\sigma$.
We then define $\hat{T}_{r}=(\hat{V},\hat{E})$ 
as a suffix tree on the set $\mathbb{P}_{\sigma}^r$, i.e.,
the vertex and edge sets of $\hat{T}_r$ are
\[
\hat{V} = \mathbb{P}_{\sigma}^r
\quad\mbox{and}\quad
\hat{E} = \{(Q,Qv) \mid Q, Qv \in \mathbb{P}_{\sigma}^r\}.
\]

The mapping is defined to be the end vertex (resp., edge) of 
each path in $\mathbb{P}_{\sigma}^r$.
\[
\begin{array}{lll}
\Psi(Qv)  &= v, 
   &\mbox{ for all $Qv\in \hat{V}$ where
           $Q\in \mathbb{P}_{\sigma}^r, v\in V(G)$}\\
\Phi(Qu,Quv) &= \phi(u,v), 
   &\mbox{ for all $(Qu,Quv)\in \hat{E}$ where
           $Q\in \mathbb{P}_{\sigma}^r; u,v\in V(G)$}\\
\hat{w}_{\hat{u}\hat{v}} &= w(\phi(u,v)),
   &\mbox{ for all $\hat{u}\hat{v}\in \hat{E}$ where
           $\psi(\hat{u}) = Qu$ and $\psi(\hat{v})=Quv$}
\end{array}
\]

Now consider any arborescence $J$ in $G$, and the arborescence $J'$
of height $\sigma$ as in Theorem~\ref{thm:height-reduction}.
It is easy to see that there is a tree $\hat{J}$ in $\hat{T}_{r}$  
such that $V(J')=\Psi(V(\hat{J}))$, and $E(J')=\Psi(E(\hat{J}))$.
This simply follows by the construction of $\hat{T}$ 
that we list all the rooted-path of length at most $\sigma$ in $\Metric{G}$.
For any edge $(u,v)\in E(J')$, where $u=\psi(\hat{u})$ and
$v=\psi(\hat{v})$, we know that 
$\hat{w}(\hat{u}\hat{v}) \leq w'(u,v)$ 
because $\hat{w}(\hat{u}\hat{v})$ is the weight of a shortest $u,v$-path in $G$
whereas $w'(u,v)$ is the weight of a $u,v$-path 
in the subgraph $J$ of $G$. 
This proves Lemma~\ref{lem:height-reduction}.

\end{document}